\documentclass{amsart}

\usepackage{amsmath,amssymb}
\usepackage{amsthm}
\usepackage{amsrefs}
\usepackage{indentfirst}
\usepackage{mathrsfs}
\usepackage{multirow}
\usepackage{graphicx}
\usepackage{hyperref}
\usepackage{xcolor}
\usepackage{fourier}
\usepackage{float}
\usepackage[norelsize,linesnumbered,vlined,ruled,algo2e]{algorithm2e} 
\theoremstyle{plain}
\newtheorem{theorem}{Theorem}
\newtheorem{lemma}{Lemma}

\newtheorem{prop}[theorem]{Proposition}

\theoremstyle{definition}
\newtheorem{defn}[theorem]{Definition}

\newtheorem{algorithm}[theorem]{{\bf Algorithm}}

\theoremstyle{remark}
\newtheorem*{remark}{Remark}
\newtheorem{example}{Example}

\DeclareMathOperator{\tr}{tr}

\DeclareMathOperator{\spec}{spec}

\newcommand{\RR}{\mathbb{R}}

\newcommand{\TT}{\mathrm{T}}

\newcommand{\wt}[1]{\widetilde{#1}}

\def\Ps{P^{*}}
\def\Qs{Q^{*}}
\def\Rs{R^{*}}
\def\bs{b^{*}}
\def\ds{d^{*}}

\IfFileExists{mathabx.sty}%
  {\DeclareFontFamily{U}{mathx}{\hyphenchar\font45}%
   \DeclareFontShape{U}{mathx}{m}{n}{<->mathx10}{}%
   \DeclareSymbolFont{mathx}{U}{mathx}{m}{n}%
   \DeclareFontSubstitution{U}{mathx}{m}{n}%
   \DeclareMathAccent{\widebar}{0}{mathx}{"73}%
}{%
  \PackageWarning{mathabx}{%
    Package mathabx not available, therefore\MessageBreak substituting
    widebar with overline\MessageBreak }%
  \newcommand{\widebar}[1]{\overline{#1}}%
}
\newcommand{\wb}[1]{\widebar{#1}}

\newcommand{\mc}[1]{\mathcal{#1}}

\newcommand{\norm}[1]{\lVert#1\rVert}

\newcommand{\x}{\mathrm{x}}

\begin{document}

\title{Density matrix minimization with $\ell_1$ regularization} 

\thanks{The research of J.L.~was supported in part by the Alfred
  P.~Sloan Foundation and the National Science Foundation under award
  DMS-1312659. The research of S.O.~was supported by the Office of
  Naval Research Grant N00014-11-1-719. }

\author{Rongjie Lai}
\address{Department of Mathematics, University of California, Irvine}

\author{Jianfeng Lu}
\address{Department of Mathematics, Physics, and Chemistry, Duke University}

\author{Stanley Osher}
\address{Department of Mathematics and Institute for Pure and Applied Mathematics, University of California, Los Angeles}

\date{\today}

\begin{abstract}
  We propose a convex variational principle to find sparse
  representation of low-lying eigenspace of symmetric matrices. In the
  context of electronic structure calculation, this corresponds to a
  sparse density matrix minimization algorithm with $\ell_1$
  regularization. The minimization problem can be efficiently solved
  by a split Bergman iteration type algorithm. We further prove that from any
  initial condition, the algorithm converges to a minimizer of the
  variational principle. 
\end{abstract}

\maketitle

\section{Introduction}

The low-lying eigenspace of operators has
many important applications, including those in quantum chemistry,
numerical PDEs, and statistics. Given a $n \times n$ symmetric matrix
$H$, and denote its eigenvectors as $\{\Phi_i\}, i = 1, \ldots,
n$. The low-lying eigenspace is given by the span of the first $N$
(usually $N \ll n$) eigenvectors.

In many scenario, the real interest is the subspace itself, but not a
particular set of basis functions. In particular, we are interested in a
sparse representation of the eigenspace. The eigenvectors form a
natural basis set, but for oftentimes they are not sparse or localized
(consider for example the eigenfunctions of the free Laplacian
operator $-\Delta$ on a periodic box). This suggests asking for an alternative
sparse representation of the eigenspace.

In quantum chemistry, the low-lying eigenspace for a Hamiltonian
operator corresponds to the physically occupied space of electrons. In
this context, a localized class of basis functions of the low-lying
eigenspaces is called Wannier functions \cites{Wannier:1937,
  Kohn:1959Wannier}. These functions provide transparent
interpretation and understandings of covalent bonds, polarizations,
\textit{etc.} of the electronic structure. These localized
representations are also the starting point and the essence for many
efficient algorithms for electronic structure calculations (see
e.g.~the review article~\cite{Goedecker:99}).

\subsection{Our contribution}

In this work, we propose a convex minimization principle for finding a
sparse representation of the low-lying eigenspace.
\begin{equation}\label{eqn:l1DensityMatrix}
  \begin{aligned}
    & \min_{P \in \RR^{n\times n}} \tr(H P) + \frac{1}{\mu} \norm{P}_1\\
    & \text{s.t.} \  P = P^{\TT},\, \tr P = N,\, 0 \preceq P \preceq I,
  \end{aligned}
\end{equation}
where $\norm{\cdot}_1$ is the entrywise $\ell_1$ matrix norm, $A
\preceq B$ denotes that $B - A$ is a positive semi-definite matrix,
and $\mu$ is a penalty parameter for entrywise sparsity. Here $H$ is an
$n\times n$ symmetric matrix, which is the (discrete) Hamiltonian in
the electronic structure context. The variational principle gives $P$
as a sparse representation of the projection operator onto the
low-lying eigenspace.

The key observation here is to use the matrix $P$ instead of the wave
functions $\Psi$. This leads to a convex variational
principle. Physically, this corresponds to looking for a sparse
representation of the density matrix. We also noted that in cases
where we expect degeneracy or near-degeneracy of eigenvalues of the
matrix $H$, the formulation in terms of the density matrix $P$ is more
natural, as it allows fractional occupation of states. This is a
further advantage besides the convexity.

Moreover, we design an efficient minimization algorithm based on split
Bregman iteration to solve the above variational problem.  Starting
from any initial condition, the algorithm always converges to a
minimizer.

\subsection{Previous works}

There is an enormous literature on numerical algorithms for
Wannier functions and more generally sparse representation of
low-lying eigenspace.  The influential work \cite{Marzari:1997}
proposed a minimization strategy within the occupied space to find
spatially localized Wannier functions (coined as ``maximally localized
Wannier functions'').

In \cite{E:2010PNAS}, the second author with his collaborators
developed a localized subspace iteration (LSI) algorithm to find
Wannier functions. The idea behind the LSI algorithm is to combine the
localization step with the subspace iteration method as an iterative
algorithm to find Wannier functions of an operator. The method has
been applied to electronic structure calculation in
\cite{Garcia-CerveraLuXuanE:09}. As \cite{Garcia-CerveraLuXuanE:09}
shows, due to the truncation step involved, the LSI algorithm does not
in general guarantee convergence.

As a more recent work in \cites{OzolinsLaiCaflischOsher:13}, $L_1$ regularization is proposed to
be used in the variational formulation of the Schr\"odinger equation
of quantum mechanics for creating compressed modes, a set of spatially
localized functions $\{\psi_i\}_{i=1}^N$ in $\mathbb{R}^d$ with
compact support.
\begin{equation}
\label{model:CMs}
E = \min_{\Psi_N} \sum_{j=1}^N  \left( \frac{1}{\mu}\left| \psi_j \right|_1 +  \langle \psi_j , \hat{H} \psi_j \rangle \right) \quad \mbox{\text{s.t.}} \quad \langle \psi_j, \psi_k \rangle = \delta_{jk},
\end{equation}
where $\hat{H} = -\frac{1}{2}\Delta + V(\x)$ is the Hamilton operator
corresponding to potential $V(\x)$, and the $L_1$ norm is defined
as $\left| \psi_j \right|_1 = \int | \psi_j | d\x$.  This $L_1$
regularized variational approach describes a general formalism for
obtaining localized (in fact, compactly supported) solutions to a
class of mathematical physics PDEs, which can be recast as variational
optimization problems. Although an efficient algorithm based on a method of splitting
orthogonality constraints (SOC)~\cite{Lai:2014splitting} is designed to solve the above non-convex problem, it is still
a big challenge to theoretically analyze the convergence of the proposed the algorithm.

The key idea in the proposed convex formulation~\eqref{eqn:l1DensityMatrix} of the variational principle is
the use of the density matrix $P$.  The density matrix is widely used
in electronic structure calculations, for example the density matrix
minimization algorithm \cite{LiNunesVanderbilt:93}. In this type of
algorithm, sparsity of density matrix is specified explicitly by
restricting the matrix to be a banded matrix. The resulting
minimization problem is then non-convex and found to suffer from many local
minimizers. Other electronic structure algorithms that use density
matrix include density matrix purification \cite{McWeeny:60}, Fermi
operator expansion algorithm \cite{BaroniGiannozzi:92}, just to name a
few.

From a mathematical point of view, the use of density matrix can be
viewed as similar to the idea of lifting, which has been recently used
in recovery problems \cite{CandesStrohmerVoroninski:13}. While a
nuclear norm is used in PhaseLift method
\cite{CandesStrohmerVoroninski:13} to enhance sparsity in terms of
matrix rank; we will use an entrywise $\ell_1$ norm to favor sparsity
in matrix entries.

\bigskip

The rest of the paper is organized as follows. We formulate and
explain the convex variational principle for finding localized
representations of the low-lying eigenspace in
Section~\ref{sec:formulation}. An efficient algorithm is proposed in
Section~\ref{sec:Algorithm} to solve the variational
principle, with numerical examples presented in Section~\ref{sec:numerics}. The convergence proof of the algorithm is given in Section~\ref{sec:proof}. 

\section{Formulation}\label{sec:formulation}

Let us denote by $H$ a symmetric matrix \footnote{With obvious changes, our results generalize to the Hermitian case} coming from, for example, the
discretization of an effective Hamiltonian operator in electronic
structure theory. We are interested in a sparse representation of the
eigenspace corresponding to its low-lying eigenvalues. In physical
applications, this corresponds to the occupied space of a Hamiltonian;
in data analysis, this corresponds to the principal components (for
which we take the negative of the matrix so that the largest
eigenvalue becomes the smallest). We are mainly interested in physics
application here, and henceforth, we will mainly interpret the formulation and algorithms from a physical view point. 

The Wannier functions, originally defined for periodic Schr\"odinger
operators, are spatially localized basis functions of the occupied
space. In \cite{OzolinsLaiCaflischOsher:13}, it was proposed to find the
spatially localized functions by minimizing the variational problem
\begin{equation}\label{eq:Psi}
  \min_{\Psi \in \RR^{n \times N},\,  \Psi^{\TT} \Psi = I} \tr (\Psi^{\TT} H \Psi) + \frac{1}{\mu} \norm{\Psi}_{1}
\end{equation}
where $\norm{\Psi}_{1}$ denotes the entrywise $\ell_1$ norm of $\Psi$. Here $N$ is the number of Wannier functions and $n$ is the number of spatial degree of freedom (e.g. number of spatial grid points or basis functions). 

The idea of the above minimization can be easily understood by looking at each term in the energy functional. The $\tr(\Psi^{\TT} H \Psi)$ is the sum of the Ritz value in the space spanned by the columns of $\Psi$. Hence, without the $\ell_1$ penalty term, the minimization 
\begin{equation}
  \min_{\Psi \in \RR^{n \times N},\, \Psi^{\TT} \Psi = I} \tr (\Psi^{\TT} H \Psi) 
\end{equation}
gives the eigenspace corresponds to the first $N$ eigenvalues (here
and below, we assume the non-degeneracy that the $N$-th and $(N+1)$-th
eigenvalues of $H$ are different). While the $\ell_1$ penalty prefers
$\Psi$ to be a set of sparse vectors. The competition of the two terms gives a sparse representation of a subspace that is close to the eigenspace. 

Due to the orthonormality constraint $\Psi^{\TT} \Psi = I$, the minimization problem \eqref{eq:Psi} is not convex, which may result in troubles in finding the minimizer of the above minimization problem and also makes the proof of convergence difficult.

Here we take an alternative viewpoint, which gives a convex
optimization problem. The key idea is instead of $\Psi$, we consider $P
= \Psi \Psi^{\TT} \in \RR^{n\times n}$. Since the columns of $\Psi$
form an orthonormal set of vectors, $P$ is the projection operator
onto the space spanned by $\Psi$.  In physical terms, if $\Psi$ are
the eigenfunctions of $H$, $P$ is then the density matrix which corresponds
to the Hamiltonian operator. For insulating systems, it is known that
the off-diagonal terms in the density matrix decay exponentially fast
\cites{Kohn:59, Panati:07, Cloizeaux:64a, Cloizeaux:64b, Nenciu:83,
  Kivelson:82, NenciuNenciu:98, ELu:CPAM, ELu:13}.

We propose to look for a sparse approximation of the exact density matrix by solving the minimization 
problem proposed in \eqref{eqn:l1DensityMatrix}.
%\begin{equation}\label{eq:P}
%  \begin{aligned}
%    & \min_{P \in \RR^{n\times n}} \tr(H P) + \frac{1}{\mu} \norm{P}_1\\
%    & \text{s.t.} \  P = P^{\TT},\, \tr P = N,\, 0 \preceq P \preceq I,
%  \end{aligned}
%\end{equation}
%where $\norm{P}_1$ is again the entrywise $\ell_1$ norm of the matrix $P$
%and we denote $A \preceq B$ if $B - A$ is a positive semi-definite
%matrix.
The variational problem \eqref{eqn:l1DensityMatrix} is a convex relaxation of the non-convex variational problem 
\begin{equation}\label{eq:Pnonconvex}
  \begin{aligned}
    & \min_{P \in \RR^{n\times n}} \tr(H P) + \frac{1}{\mu} \norm{P}_1 \\
    & \text{s.t.} \  P = P^{\TT},\, \tr P = N,\, P = P^2,
  \end{aligned}
\end{equation}
where the constraint $0 \preceq P \preceq I$ is replaced by the
idempotency constraint of $P$: $P = P^2$. The variational principle
\eqref{eq:Pnonconvex} can be understood as a reformulation of
\eqref{eq:Psi} using the density matrix as variable. The idempotency
condition $P = P^2$ is indeed the analog of the orthogonality
constraint $\Psi^{\TT} \Psi = I$. 
Note that $0 \preceq P \preceq I$ requires that the eigenvalues of $P$
(the occupation number in physical terms) are between $0$ and $1$,
while $P = P^2$ requires the eigenvalues are either $0$ or $1$. Hence,
the set 
\begin{equation}
\mathcal{C} =  \{ P: P = P^{\TT},\, \tr P = N,\, 0 \preceq P \preceq I \}
\end{equation}
is the convex hull of the set 
\begin{equation}
\mathcal{D} =   \{ P: P = P^{\TT},\, \tr P = N,\, P = P^2\}.
\end{equation}
Therefore \eqref{eqn:l1DensityMatrix} is indeed a convex relaxation of \eqref{eq:Pnonconvex}.

Without the $\ell_1$ regularization, the variational problems \eqref{eqn:l1DensityMatrix} and \eqref{eq:Pnonconvex} become
\begin{equation}\label{eq:P'}
  \begin{aligned}
    & \min_{P \in \RR^{n\times n}} \tr(H P) \\
    & \text{s.t.} \  P = P^{\TT},\, \tr P = N,\, 0 \preceq P \preceq I,
  \end{aligned}
\end{equation}
and
\begin{equation}\label{eq:Pnonconvex'}
  \begin{aligned}
    & \min_{P \in \RR^{n\times n}} \tr(H P) \\
    & \text{s.t.} \  P = P^{\TT},\, \tr P = N,\, P = P^2. 
  \end{aligned}
\end{equation}
These two minimizations actually lead to the same result in the
non-degenerate case.

\begin{prop}\label{prop:equiv}
  Let $H$ be a symmetric $n \times n$ matrix. Assume that the $N$-th
  and $(N+1)$-th eigenvalues of $H$ are distinct, the minimizers of
  \eqref{eq:P'} and \eqref{eq:Pnonconvex'} are the same.
\end{prop}
This is perhaps a folklore result in linear algebra, nevertheless we include
the short proof here for completeness.
\begin{proof}
  It is clear that the unique minimizer of \eqref{eq:Pnonconvex'} is
  given by the projection matrix on the first $N$ eigenvectors of $H$,
  given by 
  \begin{equation*}
    P_N = \sum_{i=1}^N v_i v_i^{\TT}
  \end{equation*}
  where $\{v_i\}, i = 1, \ldots, n$ are the eigenvectors of $H$,
  ordered according to their associated eigenvalues. Let us prove
  that \eqref{eq:P'} is minimized by the same solution.

  Assume $P$ is a minimizer of \eqref{eq:P'}, we calculate
  \begin{equation}\label{eq:trhp}
      \tr(HP)  = \sum_{i=1}^n v_i^{\TT} H P v_i 
      = \sum_{i=1}^n \lambda_i v_i^{\TT} P v_i
      = \sum_{i=1}^n \lambda_i \theta_i(P), 
  \end{equation}
  where $\theta_i(P) = v_i^{\TT} P v_i$. On the other hand, we have 
  \begin{equation*}
    \tr(P) = \sum_{i=1}^n v_i^{\TT} P v_i = \sum_{i=1}^n \theta_i(P) = N, 
  \end{equation*}
  and $0 \leq \theta_i(P) \leq 1$ since $0 \preceq P \preceq
  I$. Therefore, if we view \eqref{eq:trhp} as a variational problem
  with respect to $\{\theta_i\}$, it is clear that the unique minimum
  is achieved when
  \begin{equation*}
    \theta_i(P) = 
    \begin{cases}
      1, & i \leq N; \\
      0, & \text{otherwise}. 
    \end{cases}
  \end{equation*}
  We conclude the proof by noticing that the above holds if and only
  if $P = P_N$. 
\end{proof}

This result states that we can convexify the set of admissible matrices. 
We remark that, somewhat surprisingly, this result also holds for the
Hartree-Fock theory \cite{Lieb:77} which can be vaguely understood as
a nonlinear eigenvalue problem. However the resulting variational problem is still non-convex for the Hartree-Fock theory. 

Proposition~\ref{prop:equiv} implies that the variational principle \eqref{eqn:l1DensityMatrix} can be understood as an $\ell_1$ regularized version of the variational problem \eqref{eq:Pnonconvex'}. The equivalence no longer holds for \eqref{eqn:l1DensityMatrix} and \eqref{eq:Pnonconvex} with the $\ell_1$ regularization. The advantage of \eqref{eqn:l1DensityMatrix} over \eqref{eq:Pnonconvex} is that the former is a convex problem while the latter is not. 

Coming back to the properties of the variational problem \eqref{eqn:l1DensityMatrix}. We note that while the objective function of \eqref{eqn:l1DensityMatrix} is convex, it is not strictly convex as the $\ell_1$-norm is not strictly convex and the trace term is linear. Therefore, in general, the minimizer of \eqref{eqn:l1DensityMatrix} is not unique. 

\begin{example} Let $\mu \in \RR_+$, $N = 1$ and 
\begin{equation}
  H =
  \begin{pmatrix}
    1 & 0 \\
    0 & 1
  \end{pmatrix}, 
\end{equation}
The non-uniqueness comes from the degeneracy of the Hamiltonian
eigenvalues. Any diagonal matrix $P$ with trace $1$ and non-negative
diagonal entries is a minimizer.
\end{example}

\begin{example} Let $\mu = 1$, $N = 1$ and
\begin{equation}
  H = 
  \begin{pmatrix}
    1 & 0 & 0 \\
    0 & 2 & 2 \\
    0 & 2 & 2
  \end{pmatrix}
\end{equation}
The non-uniqueness comes from the competition between the trace term
and the $\ell_1$ regularization. The eigenvalues of $H$ are $0, 1$ and
$4$. Straightforward calculation shows that 
\begin{equation}
  P_0 = 
  \begin{pmatrix} 
    0 & 0 & 0 \\
    0 & 1/2 & -1/2 \\
    0 & -1/2 & 1/2 
  \end{pmatrix}
\end{equation}
which corresponds to the eigenvector $(0, \sqrt{2}/2,
-\sqrt{2}/2)^{\TT}$ associated with eigenvalue $0$ and
\begin{equation}
  P_1 = 
  \begin{pmatrix} 
    1 & 0 & 0 \\
    0 & 0 & 0 \\
    0 & 0 & 0 
  \end{pmatrix}
\end{equation}
which corresponds to the eigenvector $(1, 0, 0)^{\TT}$ associated with
eigenvalue $1$ are both minimizers of the objective function $\tr(HP)
+ \norm{P}_{1}$. Actually, due to convexity, any convex combination of
$P_0$ and $P_1$ is a minimizer too.

\end{example}

It is an open problem under what assumptions that the uniqueness is guaranteed.

\section{Algorithm}\label{sec:Algorithm}

%We solve the optimization problem (\ref{eq:P}) using split Bregman iteration~\cite{Goldstein:2009split}. 
To solve the proposed minimization problem \eqref{eqn:l1DensityMatrix}, we design a fast algorithm based on split Bregman iteration~\cite{Goldstein:2009split}, which comes from the ideas of variables splitting and Bregman iteration~\cite{Osher:2005}. Bregman iteration has attained intensive attention due to its efficiency in many $\ell_1$ related constrained optimization problems~\cite{Yin:2008bregman,yin2013error}. With the help of auxiliary variables, split Bregman iteration iteratively approaches the original optimization problem by computation of several easy-to-solve subproblems. This algorithm  popularizes the idea of using operator/variable splitting to solve optimization problems arising from information science. The equivalence of the split Bregman iteration to the alternating direction method of multipliers (ADMM), Douglas-Rachford splitting and augmented Lagrangian method can be found in \cite{Esser:2009CAM,Setzer:2009SSVMCV,Wu:2010SIAM}. 
%Moreover, we also show the theoretical convergence analysis of the proposed algorithm inspired by the general discussion of the augmented Lagrangian method in~\cite{GlowinskiLeTallec:89}. 

By introducing auxiliary variables $Q = P$ and $R=P$, the optimization problem \eqref{eqn:l1DensityMatrix} is equivalent to 
\begin{equation}\label{eq:P_split}
  \begin{aligned}
    & \min_{P, Q, R \in \RR^{n\times n}} \frac{1}{\mu} \norm{Q}_1 + \tr(H P) \\
    & \text{s.t.} \  Q = P,\, R = P,\,  \tr P = N,\,R = R^{\TT},\, 0 \preceq R \preceq I,
  \end{aligned}
\end{equation}
which can be iteratively solved by:
\begin{align}
  (P^k,Q^k,R^k) &= \arg\min_{P, Q, R \in \RR^{n\times n}} \frac{1}{\mu} \norm{Q}_1 + \tr(H P) + \frac{\lambda}{2} \|P - Q + b\|_F^2 + \frac{r}{2}\|P - R + d\|^2_F  \\
  & \qquad \text{s.t.} \qquad \tr P = N,\,R = R^{\TT},\, 0 \preceq R \preceq I,   \nonumber  \label{eqn:PQR}\\
  b^k &= b^{k-1} + P^k - Q^k   \\
  d^k &= d^{k-1} + P^k - R^k
\end{align}
%\jl{Probably need some more motivation / explanation here}
where variables $b, d$ are essentially Lagrangian multipliers and parameters $r, \lambda$ control the penalty terms.
Solving $P^k, Q^k, R^k$ in \eqref{eqn:PQR} alternatively, we have the following algorithm.
\begin{algorithm}
\label{alg:CM_P}
Initialize $Q^0 = R^0 = P^0 \in\mathcal{C} , b^0 = d^0 = 0$

\While{``not converge"}{
\begin{enumerate}
\item $\displaystyle P^k = \arg\min_{P \in \RR^{n\times n}}\tr(H P) + \frac{\lambda}{2} \|P - Q^{k-1} + b^{k-1}\|_F^2 + \frac{r}{2}\|P - R^{k-1} + d^{k-1}\|^2_F, ~ \text{s.t.} ~ \tr P = N $.
\item $\displaystyle Q^k = \arg\min_{Q  \in \RR^{n\times n}} \frac{1}{\mu}\norm{Q}_1+ \frac{\lambda}{2} \|P^k - Q + b^{k-1}\|_F^2  $.
\item $R^k = \displaystyle\arg\min_{R  \in \RR^{n\times n}}   \frac{r}{2}\|P^{k} - R + d^{k-1}\|^2_F, \quad \text{s.t.} \quad    R = R^{\TT},\, 0 \preceq R \preceq I$.
\item $b^{k} = b^{k-1} +  P^k - Q^k$.
\item $d^{k} = d^{k-1} +  P^k - R^k$.
\end{enumerate}
}
\end{algorithm}

Note that the minimization problem in the steps of Algorithm~\ref{alg:CM_P} can be solved explicitly, as follows:
\begin{align}
  P^k &= B^k - \frac{\tr (B^k) - N}{n}, \\
  & \quad  \text{where } \quad B^k = \frac{\lambda}{\lambda+r}(Q^{k-1} - b^{k-1}) + \frac{r}{\lambda+r}(R^{k-1} - d^{k-1}) - \frac{1}{\lambda +r} H   \notag \\
  Q^k &= \operatorname{Shrink}\left(P^k +b^{k-1},\frac{1}{\lambda\mu} %\frac{\lambda}{\mu}
\right) 
= \operatorname{sign}(P^k + b^{k-1})\max\left\{|P^{k} +b^{k-1}| - \frac{1}{\lambda\mu},0\right\}                \\
  R^k &= V\min\{\max\{D,0\},1\}V^T , \text{ where } [V,~ D] =
  \operatorname{eig}(P^k + d^{k-1}).
\end{align}

Starting form any initial guess, the following theorem guarantees that the algorithm converges to one of the minimizers of the variational problem \eqref{eq:P_split}. 
\begin{theorem}\label{thm:conv}
  The sequence $\big\{(P^k, Q^k, R^k)\big\}_k$ generated by
  Algorithm~\ref{alg:CM_P} from any starting point converges to a
  minimum of the variational problem \eqref{eq:P_split}.
\end{theorem}

We will prove a slightly more general version of the above
(Theorem~\ref{thm:three}). The idea of the proof follows from the
general framework of analyzing split Bregman iteration,
\textit{i.e.}~alternating direction method of multipliers (ADDM), see
for example \cite{GlowinskiLeTallec:89}. The standard proof needs to
be generalized to cover the current case of ``two level splitting''
and the non-strictly convexity of the functionals. We defer the
detailed proof to Section~\ref{sec:proof}.

\section{Numerical results}\label{sec:numerics}

In this section, numerical experiments are presented to demonstrate
the proposed model \eqref{eqn:l1DensityMatrix} for density matrix
computation using algorithm \ref{alg:CM_P}. We illustrate our
numerical results in three representative cases, free electron model,
Hamiltonian with energy band gap and a non-uniqueness example of the
proposed optimization problem. All numerical experiments are
implemented by \textsf{MATLAB} in a PC with a 16G RAM and a 2.7 GHz CPU.

\subsection{1D Laplacian}
In the first example, we consider the proposed model for the free
electron case, in other words, we consider the potential free
Schr\"{o}dinger operator $-1/2\Delta$ defined on 1D domain $\Omega =
[0,~ 100]$ with periodic boundary condition.  This model approximates
the behavior of valence electrons in a metallic solid with weak atomic
pseudopotentials.
In this case, the matrix $H$ is a central difference discretization of
$-1/2\Delta$ on $[0, ~100]$ with equally spaced $256$ points, and we
take $N = 10$.
% Let $\{\phi_1,\cdots,\phi_N\}$ be the first $N$ eigenfunctions of the operator $H$.
Figure~\ref{fig:DensityMatrix_Lap}(a) illustrates the true density
matrix $\sum_{i=1}^{10} |\phi_i\rangle \langle \phi_i |$ obtained by
the first $10$ eigenfunctions of $H$. As the free Laplacian does not
have a spectral gap, the density matrix decays slowly in the
off-diagonal direction.  Figure~\ref{fig:DensityMatrix_Lap}(b) and
(c) plot the density matrices obtained from the proposed model with
parameter $\mu = 10$ and $100$. Note that they are much localized than
the original density matrix. As $\mu$ gets larger, the variational
problem imposes a smaller penalty on the sparsity, and hence the
solution for $\mu = 100$ has a wider spread than that for $\mu = 10$.

\begin{figure}[ht]
\centering
\includegraphics[width=.5\linewidth]{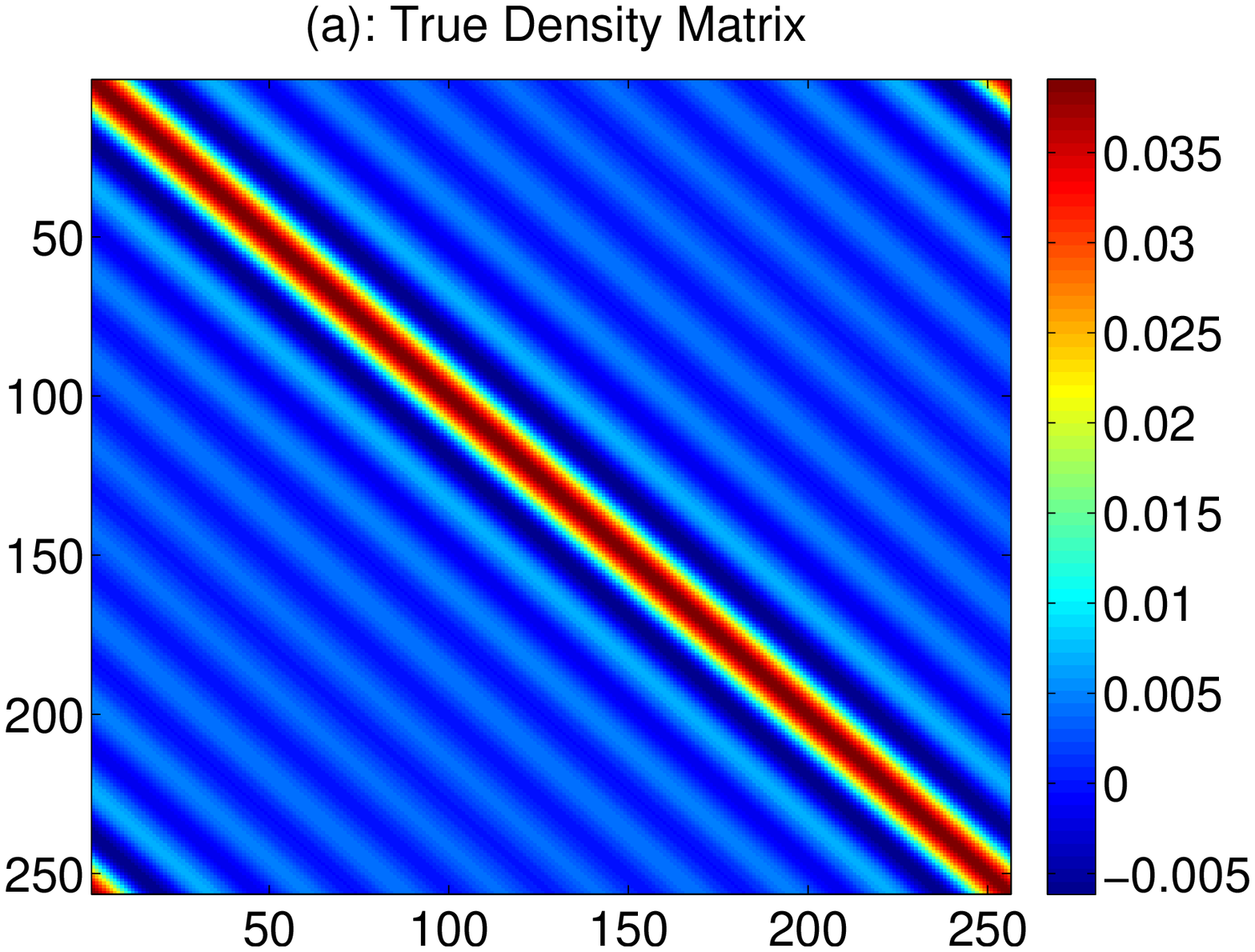}\\
\begin{minipage}{0.49\linewidth}
\includegraphics[width=1\linewidth]{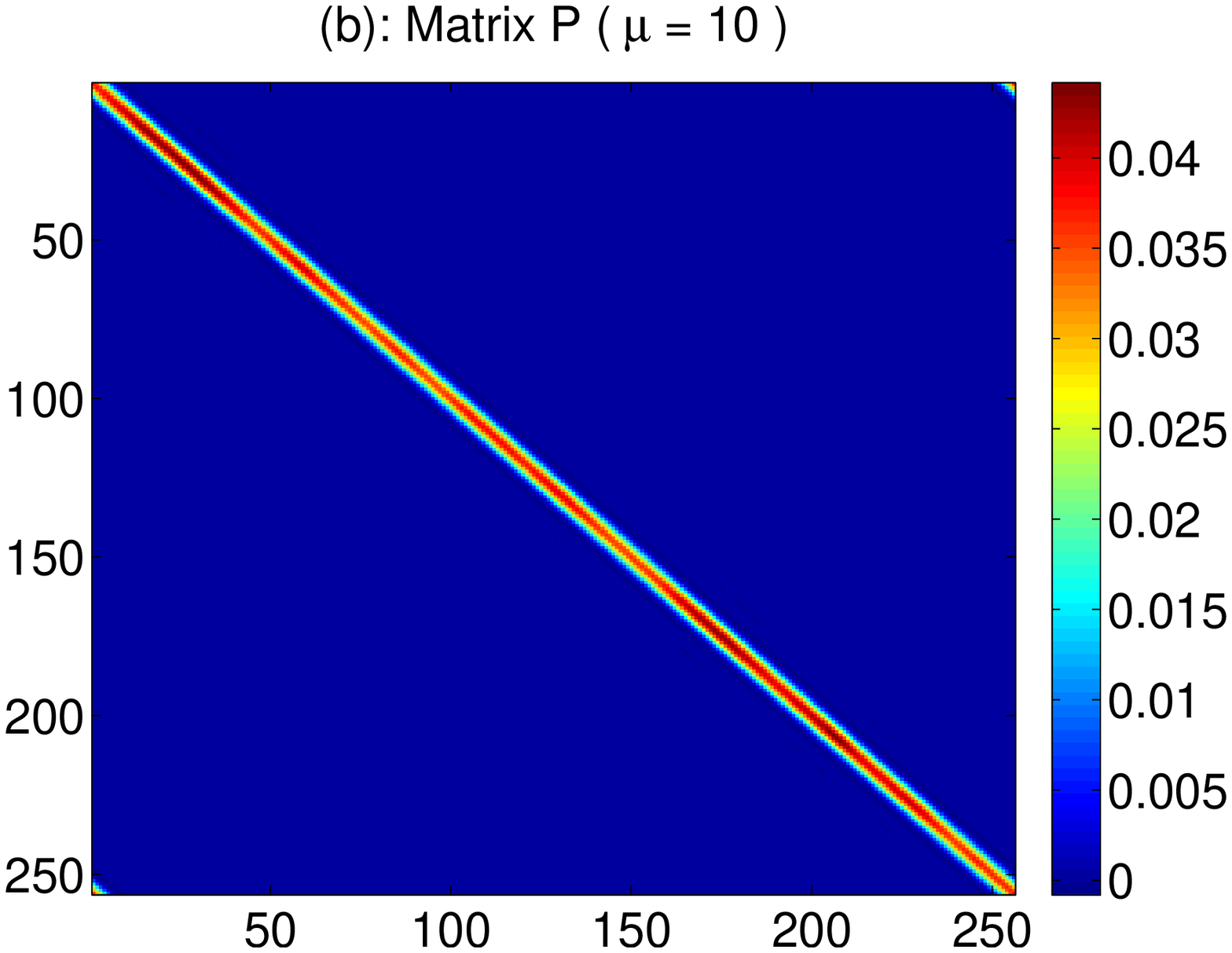}\\
\end{minipage}\hfill
\begin{minipage}{0.49\linewidth}
\includegraphics[width=1\linewidth]{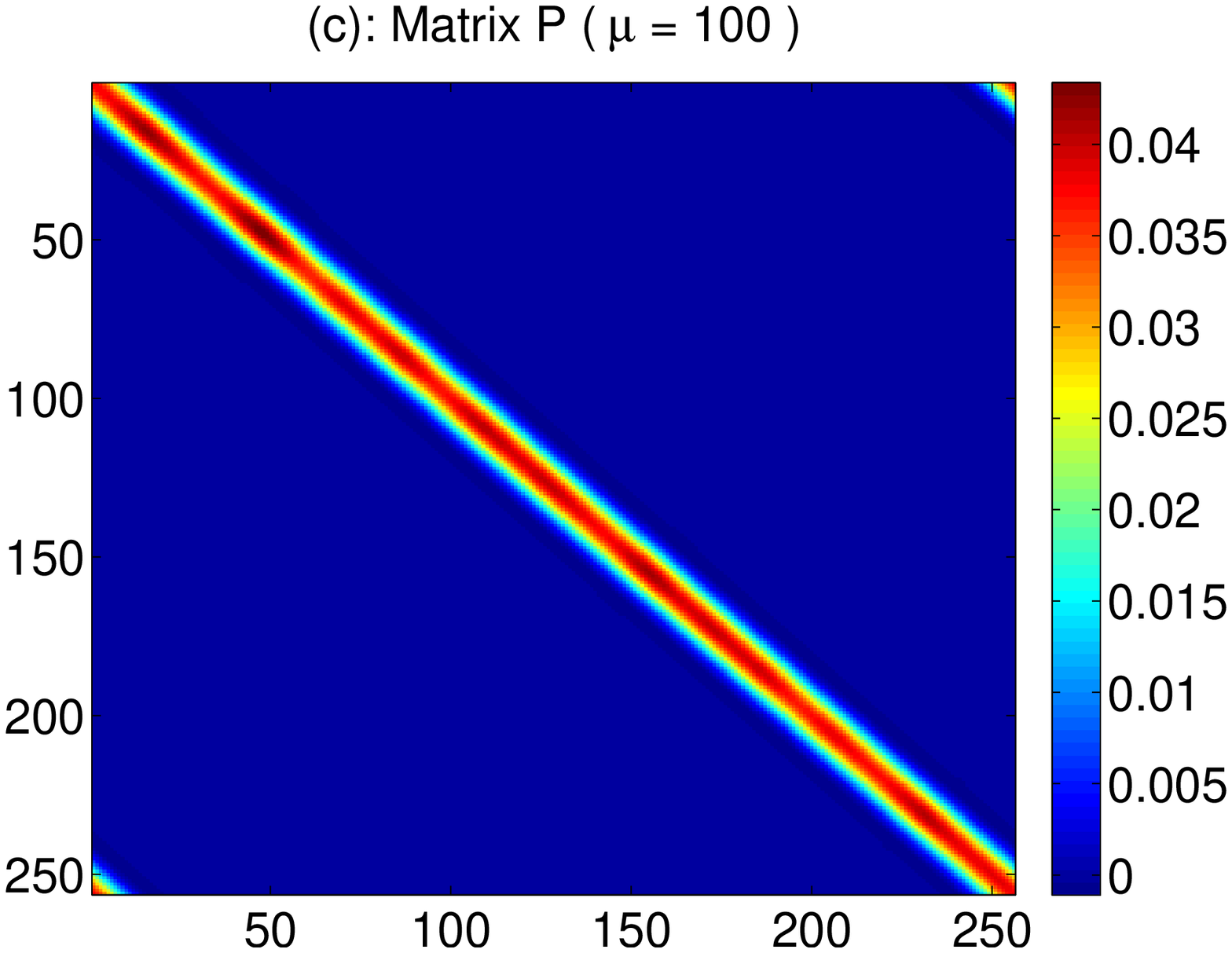}
\end{minipage}\hfill
\caption{(a): The true density matrix obtained by the first 10
  eigenfunctions of $H$. (b), (c): solutions of the density matrices
  with $\mu = 10, 100$ respectively. \label{fig:DensityMatrix_Lap}}
%\jl{Change the title of (b) and (c) to ``Matrix P ($\mu = 10$)'', etc.? One can absorb the ``(a), (b)'' etc into the title (also apply to later figures)}
\end{figure}

After we obtain the sparse representation of the density
matrix $P$, we can find localized Wannier functions as its action on
the delta functions, as plotted in Figure~\ref{fig:Projection_Lap} upper and lower pictures for $\mu = 10$ and $100$ respectively. 

\begin{figure}[htbp]
\includegraphics[width=.8\linewidth]{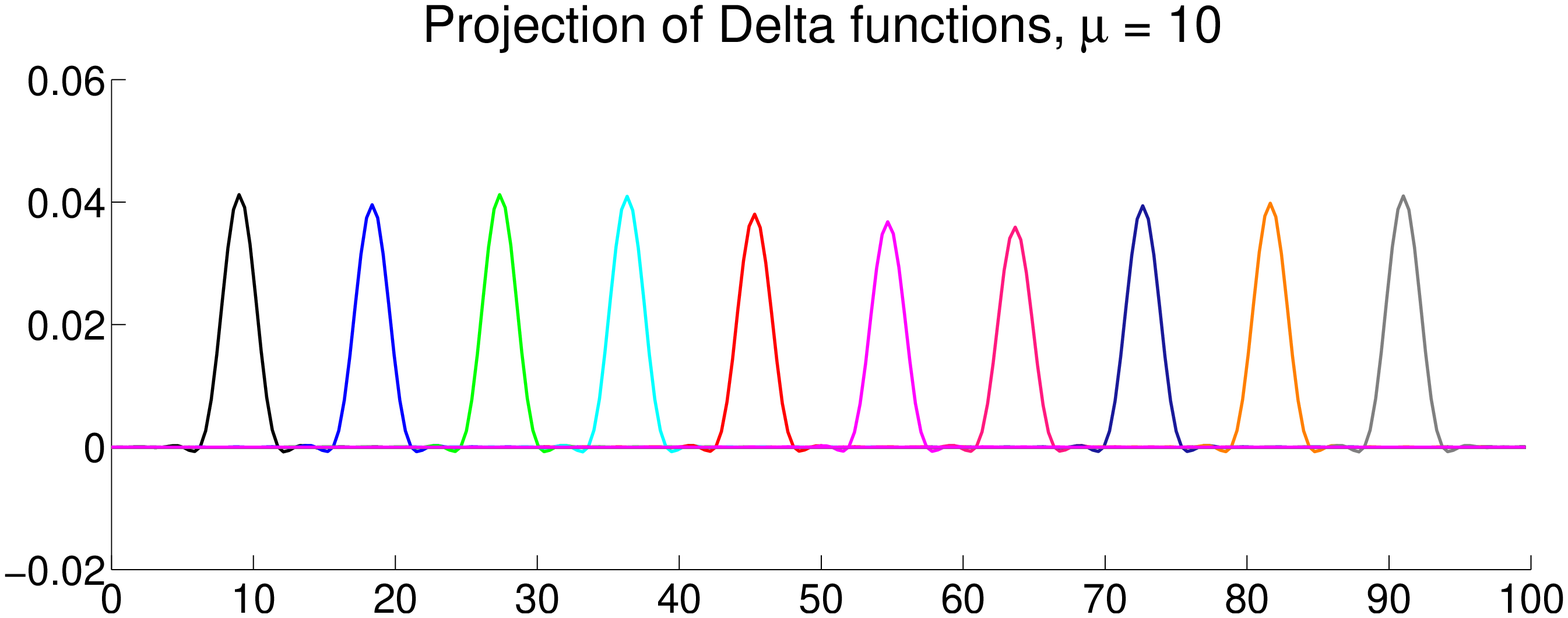}\\
\vspace{0.2cm}
\includegraphics[width=.8\linewidth]{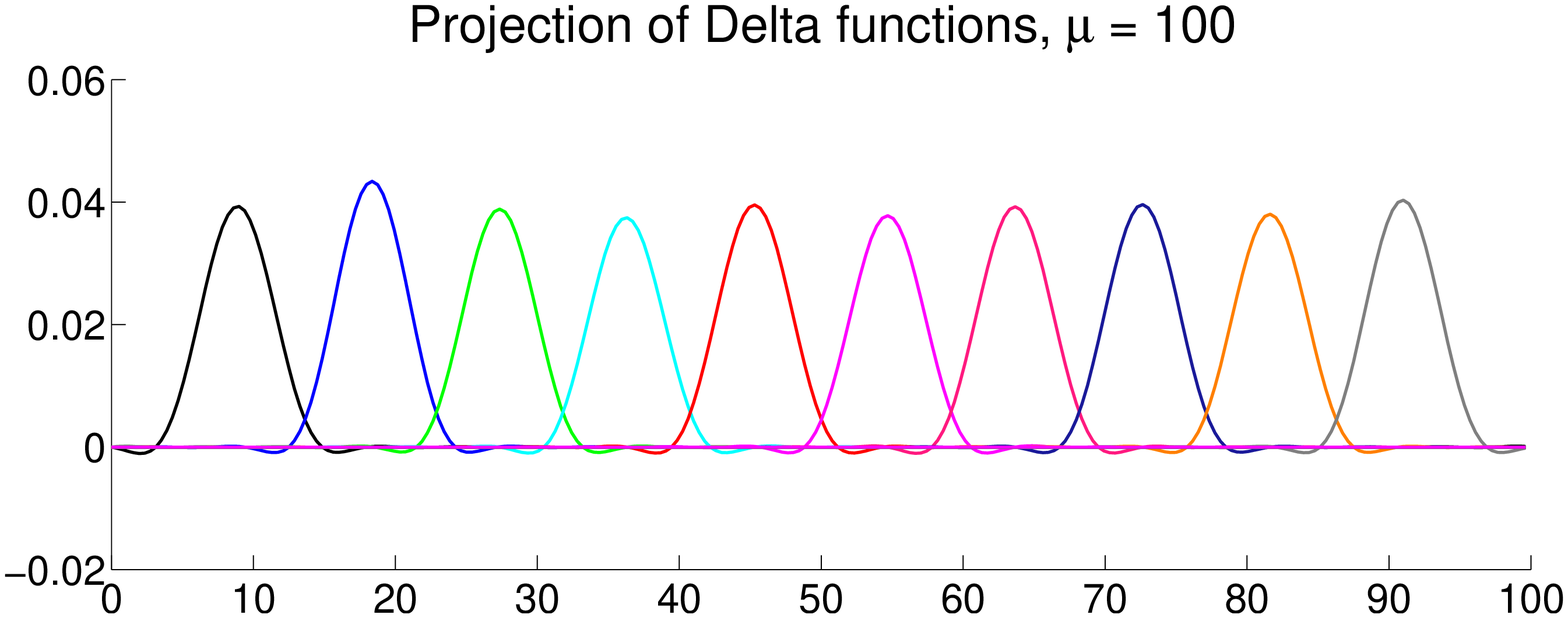}
\caption{Projection of Delta function $\delta(x - x_i)$ using density
  matrices with $\mu = 10$ (upper) and $\mu = 100$ (lower)
  respectively.}
 \label{fig:Projection_Lap}
\end{figure}

To indicate the approximation behavior of the proposed model, we
consider the energy function approximation of $ \frac{1}{\mu}
\norm{P}_1 + \tr(H P) $ to $\sum_{i=1}^{10} \langle \phi_i | H |
\phi_i \rangle$ with different values of $\mu$. In addition, we define
$\sum_{i=1}^{10} \| \phi_i - P \phi_i\|^2$ as a measurement for the
space approximation of the density matrix $P$ to the lower eigen-space
$Span\{\phi_i\}_{i=1}^{10}$.  Figure~\ref{fig:DensityFunApprox_Lap}
reports the energy approximation and the space approximation with
different values of $\mu$. Both numerical results suggest that the
proposed model will converge to the energy states of the
Schr\"{o}dinger operator. We also remark that even though the exact
density matrix is not sparse, a sparse approximation gives fairly good
results in terms of energy and space approximations.

\begin{figure}[h]
\centering
\includegraphics[width=.8\linewidth]{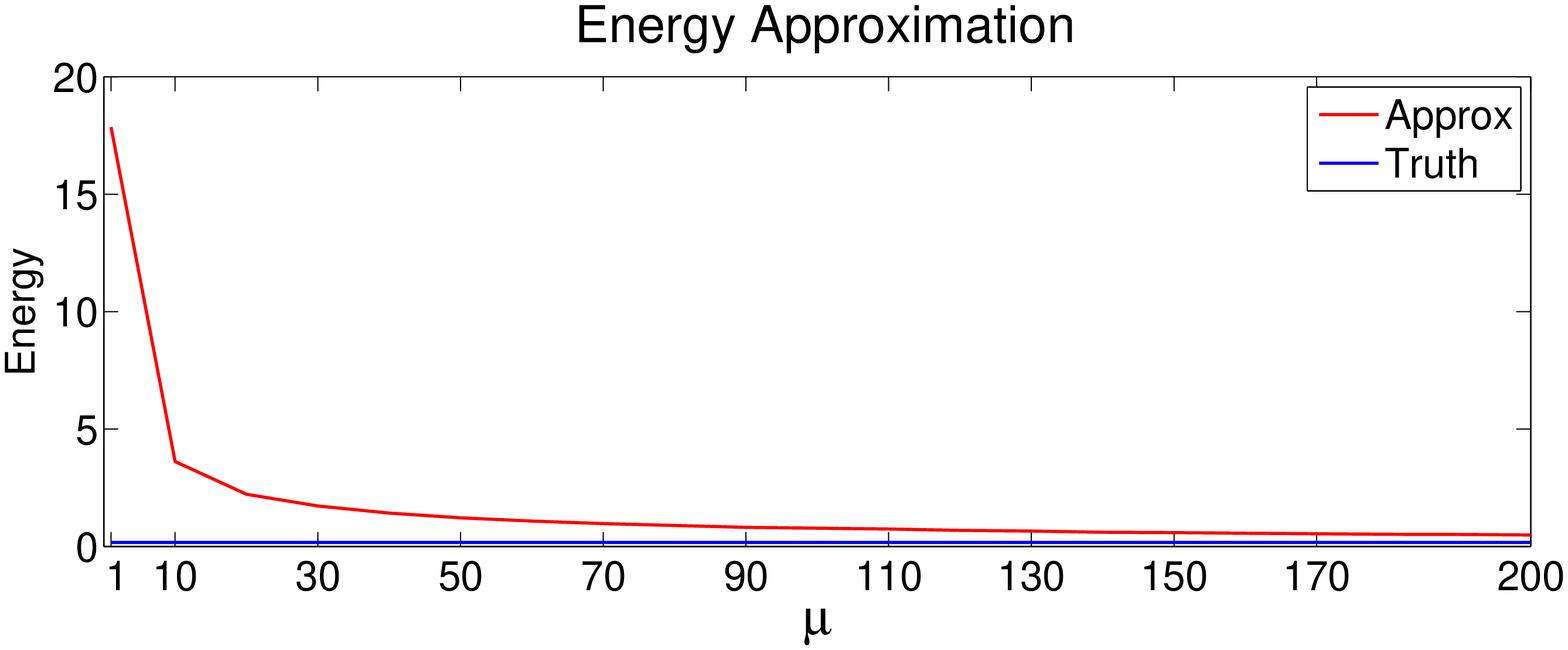}\\
\includegraphics[width=.8\linewidth]{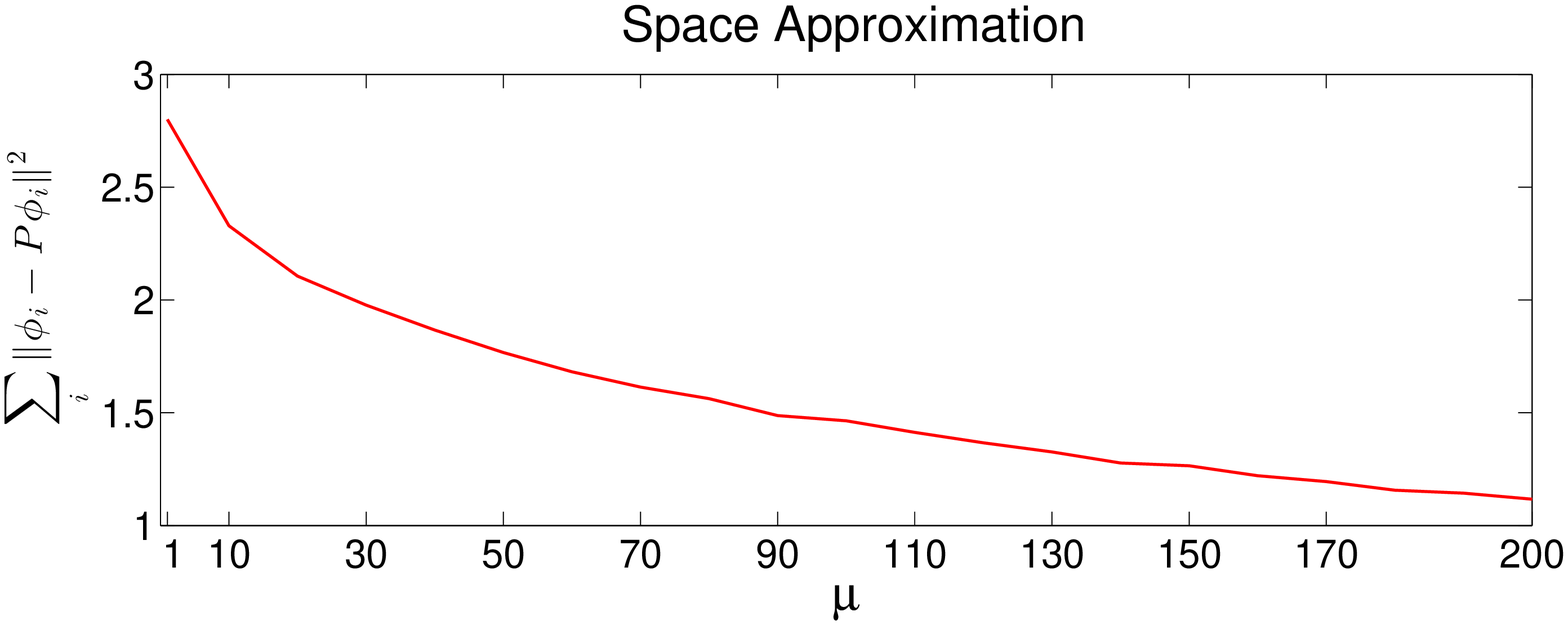}\\
\caption{Upper: Energy approximation as a function of $\mu$. Lower: Space
  approximation as a function of $\mu$.}
\label{fig:DensityFunApprox_Lap}
\end{figure}

%\begin{table}[h]
%\begin{center}
%\begin{tabular}{|c||c|c|c|c|c|c|c|}
%\hline
%\multirow{2}{1.5cm}{$\#$ of points (n)} & \multirow{2}{1.5cm}{methods } & \multicolumn{3}{|c|}{Free Electron} & \multicolumn{3}{|c|}{KP Model} \\
% \cline{3-8}
% &  &  $N = 10$ & $N= 20$  & $N = 30$  &  $N = 10$  &  $N = 20$   & $N = 30$    \\
% \cline{3-8}
% \hline
%\multirow{2}{1.5cm}{256} &  CMs   &   & &    &   0.76  & 10.82  &  30.77  \\
% \cline{2-8}
%&   proposed   &  11.98  &   23.60 &  15.33  &   &   & \\
% \hline
% \multirow{2}{1.5cm}{512} &   CMs &     &         &     &   &   & \\
% \cline{2-8}
%&   proposed   & 37.74  &  147.42  &    &    &    & \\
% \hline
%\end{tabular}
%\end{center}
%\caption{Comparisons of time consumption (seconds).}
%\label{tab:TimeComp}
%\end{table}

\subsection{1D Hamiltonian operator with a band gap}
We then consider a modified Kronig--Penney (KP)
model~\cite{Kronig:1931quantum} for a one-dimensional insulator.  The
original KP model describes the states of independent electrons in a
one-dimensional crystal, where the potential function $V(x)$ consists
of a periodic array of rectangular potential wells. We replace the
rectangular wells with inverted Gaussians so that the potential is
given by 
\begin{equation*}
  V(x) = -V_0\sum_{j=1}^{N_{\text{at}}} \exp\left[ -\frac{(x -
      x_j)^2}{\delta^2} \right], 
\end{equation*}
where $N_{\text{at}}$ gives the number of potential wells. In our
numerical experiments, we choose $N_{\text{at}} = 10$ and $x_j = 100 j
/ 11$ for $j = 1, \ldots, N_{\text{at}}$, and the domain is $[0, 100]$
with periodic boundary condition. The potential is plotted in
Figure~\ref{fig:V_KP}(a). For this given potential, the Hamiltonian
operator $H = - \tfrac{1}{2} \Delta + V(x)$ exhibits two low-energy
bands separated by finite gaps from the rest of the eigenvalue
spectrum (See Figure~\ref{fig:V_KP}(b)). Here a centered difference is
used to discretize the Hamiltonian operator.

\begin{figure}[ht]
\centering
\begin{minipage}{0.6\linewidth}
\includegraphics[width=1\linewidth]{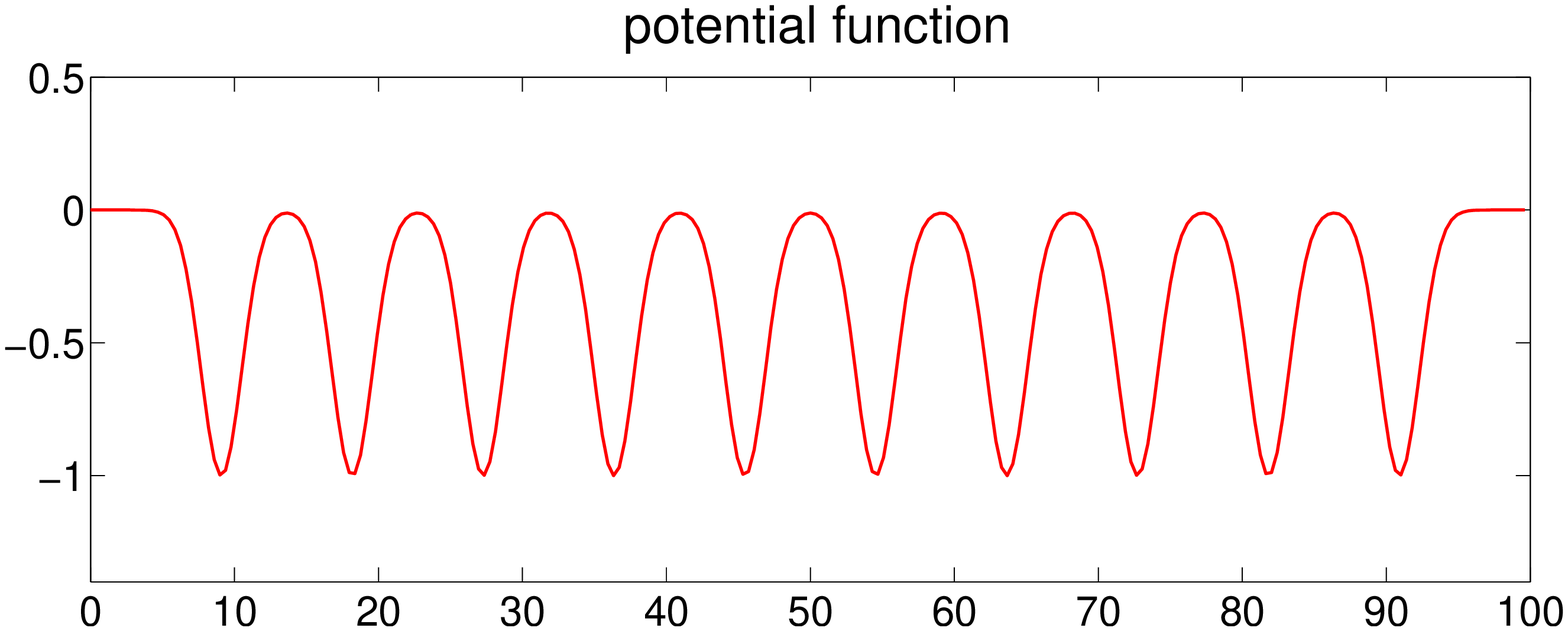}\\
\end{minipage}\hfill
\begin{minipage}{0.39\linewidth}
\includegraphics[width=.9\linewidth]{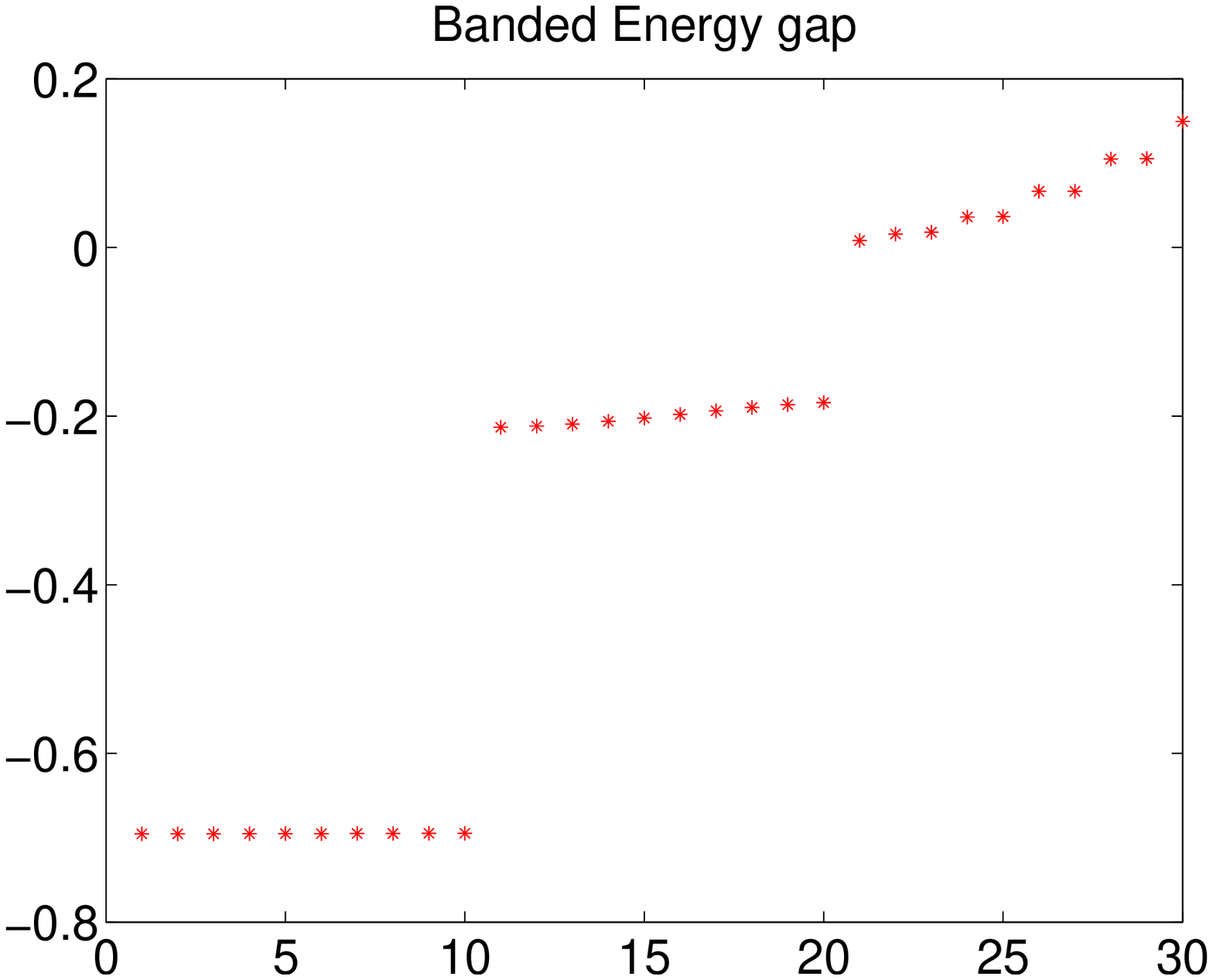}\\
\end{minipage}\hfill\\
\begin{minipage}{0.6\linewidth}
\centering (a)
\end{minipage}\hfill
\begin{minipage}{0.39\linewidth}
\centering (b)
\end{minipage}
\caption{(a): The potential function in the modified Kronig-Penney
  model. (b): The spectrum of the (discretized) Hamiltonian operator.}
\label{fig:V_KP}
\end{figure}

We consider three choices of $N$ for this model: $N = 10$, $N = 15$
and $N = 20$. They correspond to three interesting physical situations
of the model, as explained below. 

For $N = 10$, the first band of the Hamiltonian is occupied, and hence
the system has a spectral gap between the occupied and unoccupied
states. As a result, the associated density matrix is exponentially
localized, as shown in Figure~\ref{fig:DensityFunApprox_KP}(a). The
resulting sparse representation from the convex optimization is shown
in Figure~\ref{fig:DensityFunApprox_KP}(b) and (c) for $\mu = 10$ and
$100$ respectively. We see that the sparse representation agrees well
with the exact density matrix, as the latter is very localized.  The
Wannier functions obtained by projection of delta functions are shown
in Figure~\ref{fig:Projection_KP}. As the system is an
insulator, we see that the localized representation converges quickly
to the exact answer when $\mu$ increases. This is further confirmed in
Figure~\ref{fig:DensityFunApprox_KP_energy} where the energy
corresponding to the approximated density matrix and space
approximation measurement $\sum_{i=1}^{10} \| \phi_i - P \phi_i\|^2$
are plotted as functions of $\mu$.

\begin{figure}[ht]
\centering
\includegraphics[width=.5\linewidth]{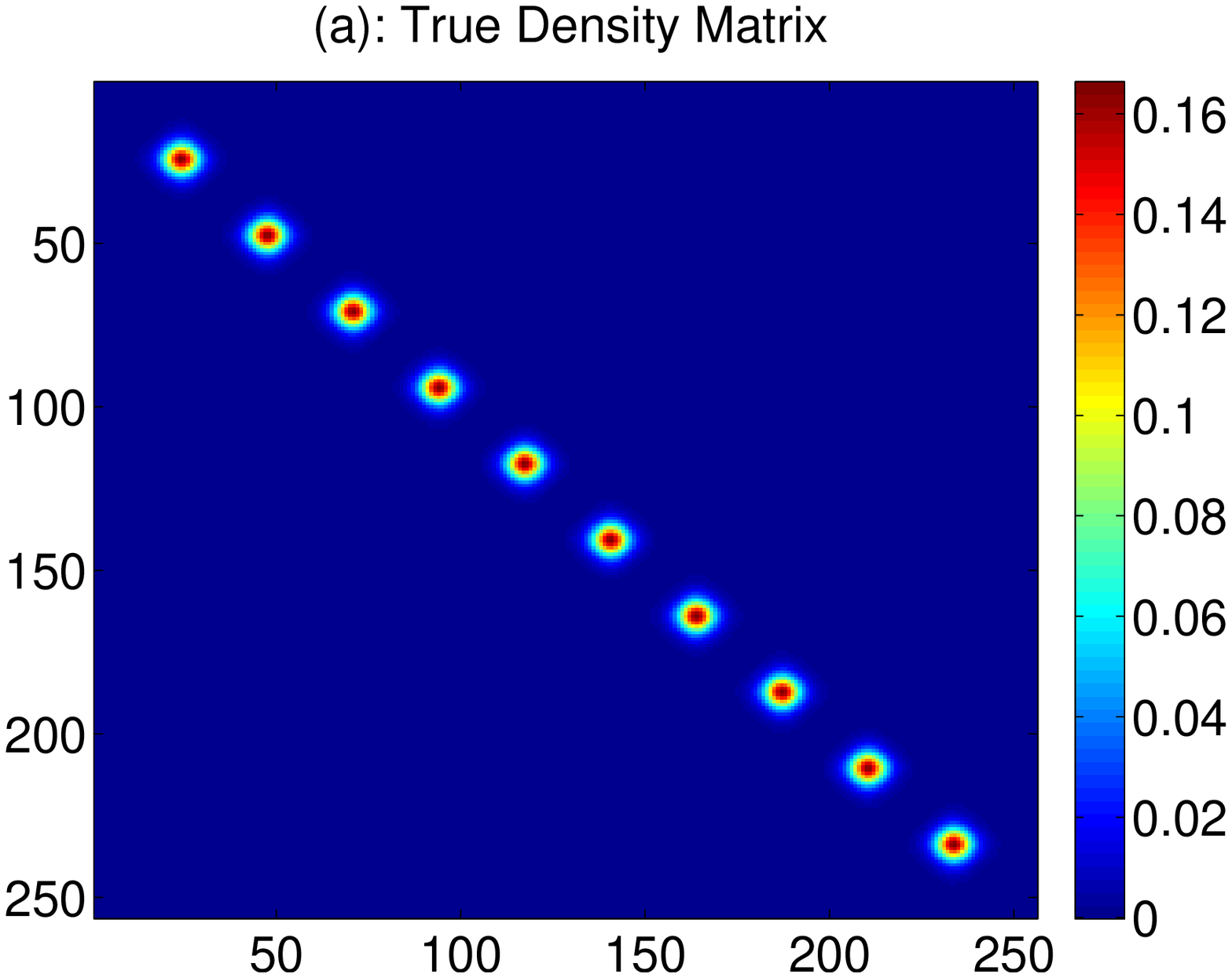}\\
\begin{minipage}{0.49\linewidth}
\includegraphics[width=1\linewidth]{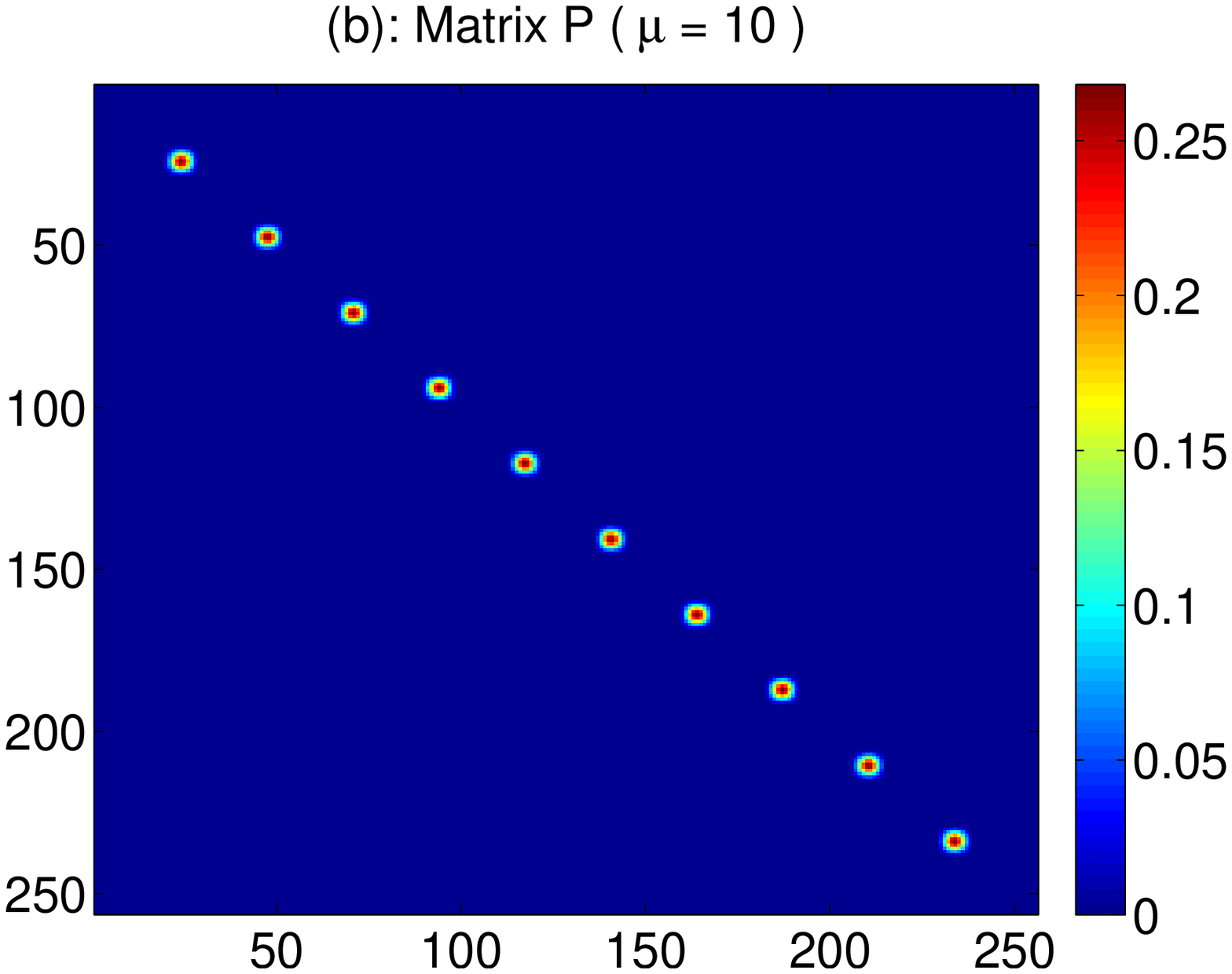}\\
\end{minipage}\hfill
\begin{minipage}{0.49\linewidth}
\includegraphics[width=1\linewidth]{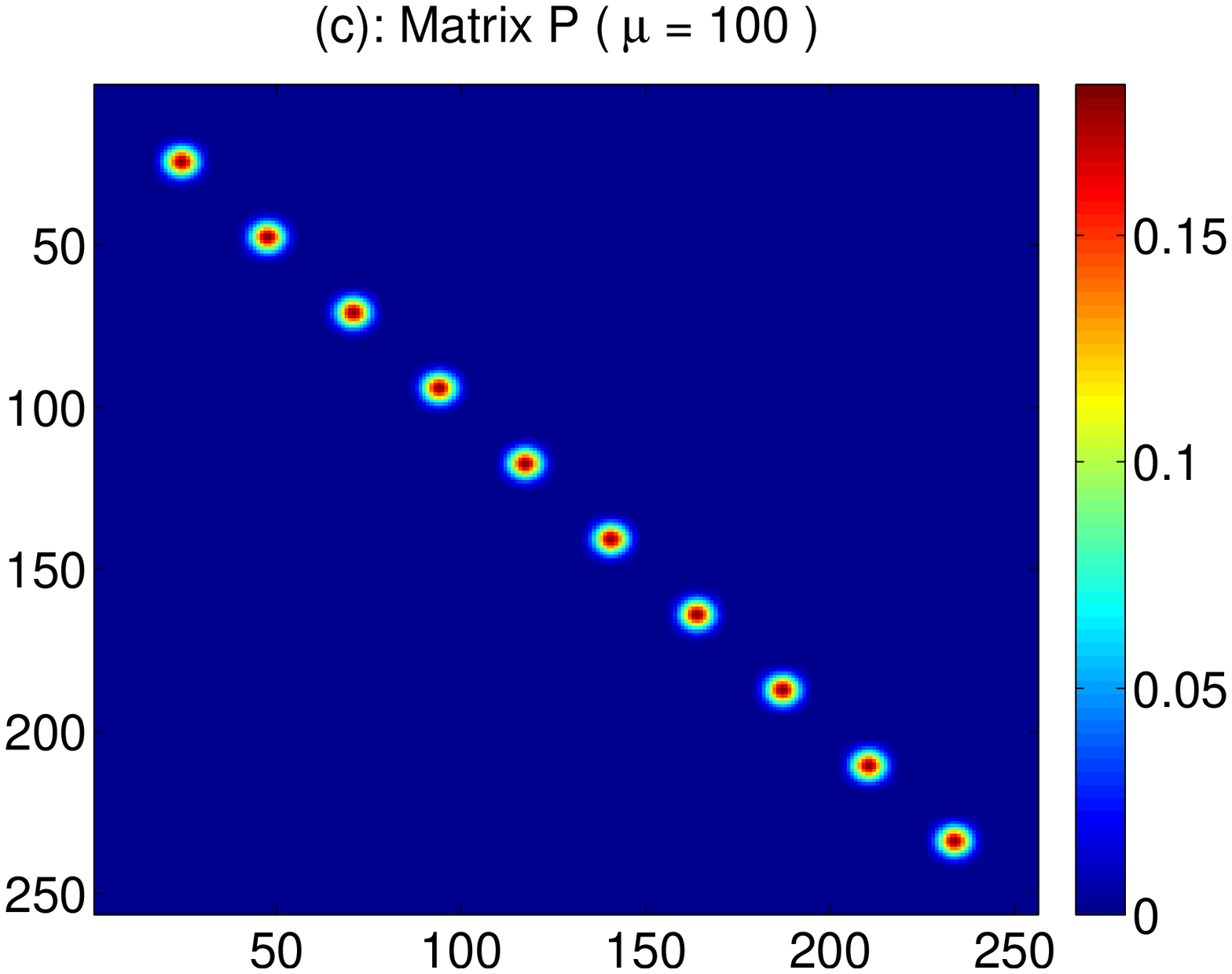}\\
\end{minipage}\hfill
\caption{(a): The true density matrix obtained by the first 10 eigenfunctions of $H$. (b), (c): solutions of the density matrices with $\mu = 10, 100$ respectively.}
\label{fig:DensityFunApprox_KP}
\end{figure}

\begin{figure}[htbp]
%\begin{minipage}{0.49\linewidth}
\centering
\includegraphics[width=.7\linewidth]{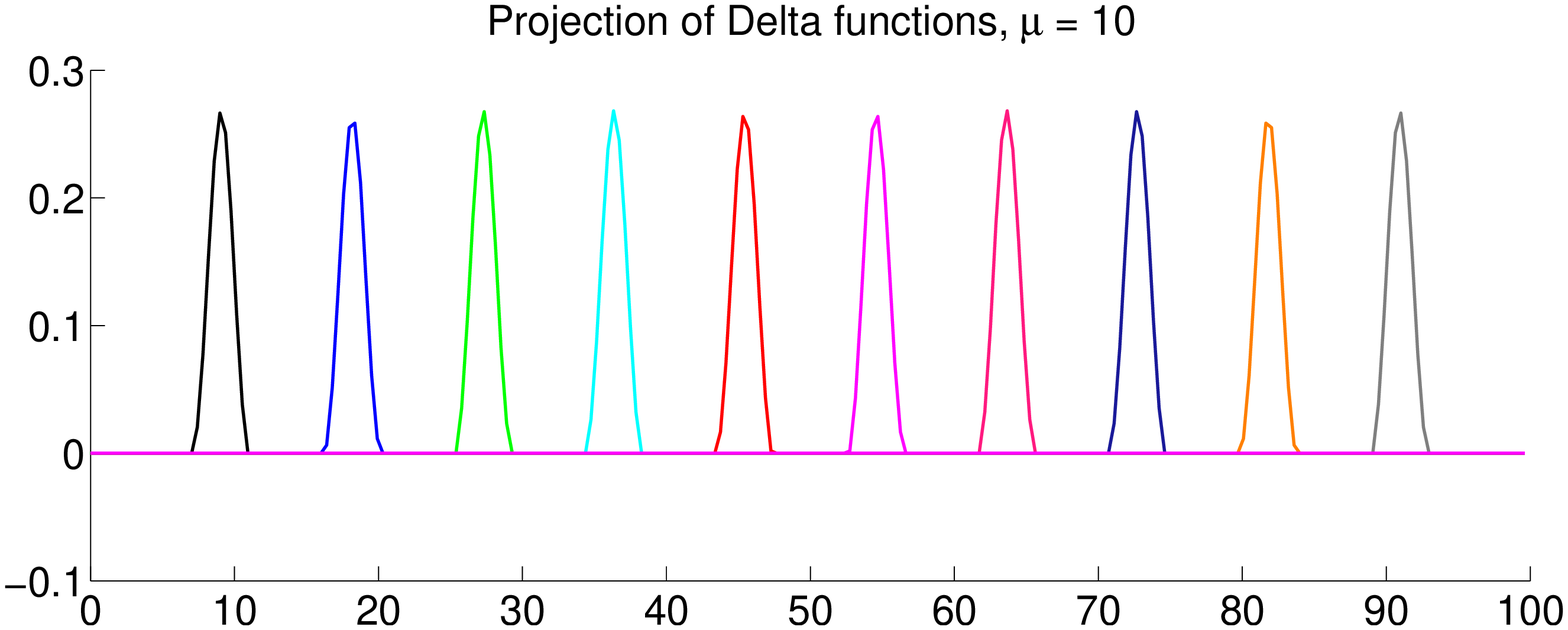}\\
%\end{minipage}\hfill
%\begin{minipage}{0.49\linewidth}
\centering
\includegraphics[width=.7\linewidth]{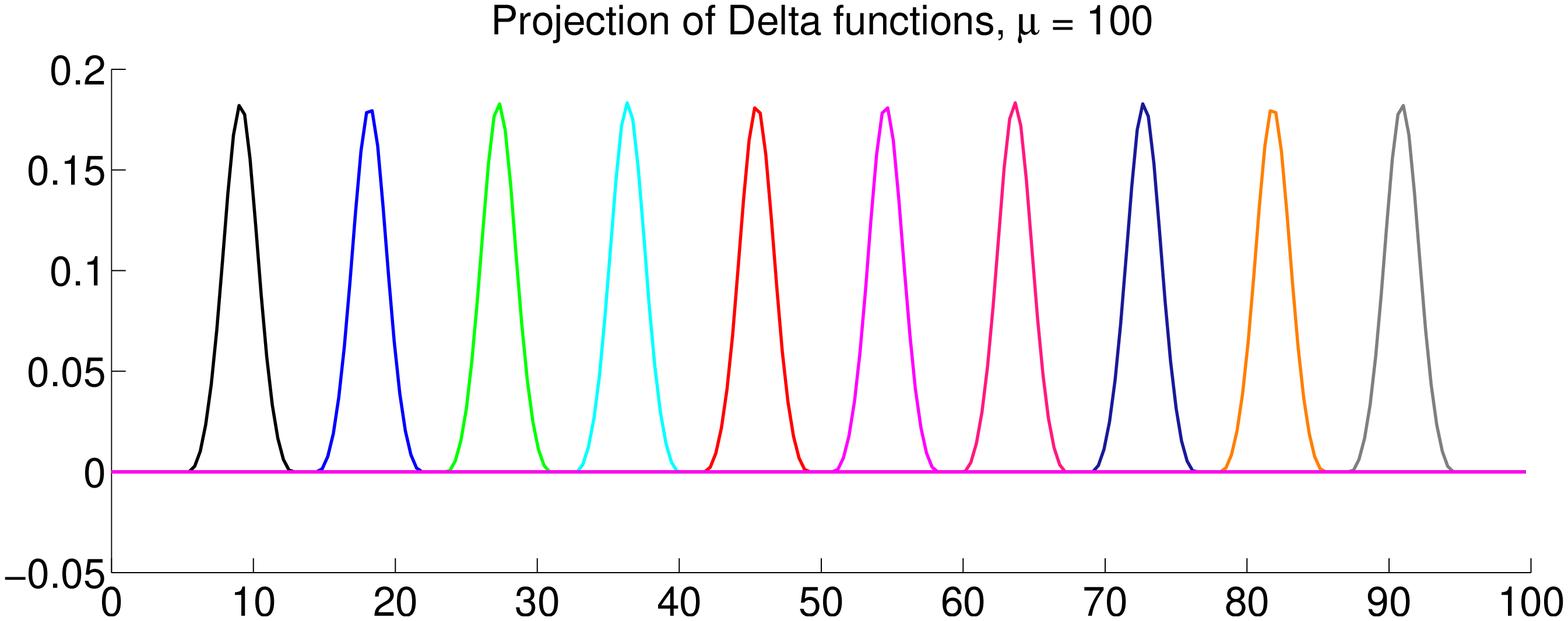}\\
%\end{minipage}\hfill
%\includegraphics[width=.8\linewidth]{KP_EnergyFunApprox.eps}\\
%\centering (f)
\caption{Projection of Delta function $\delta(x - x_i)$ using density
  matrices with $\mu = 10$ (upper) and $\mu = 100$ (lower)
  respectively.} 
\label{fig:Projection_KP}
\end{figure}

\begin{figure}[ht]
%\end{minipage}\hfill
\includegraphics[width=.8\linewidth]{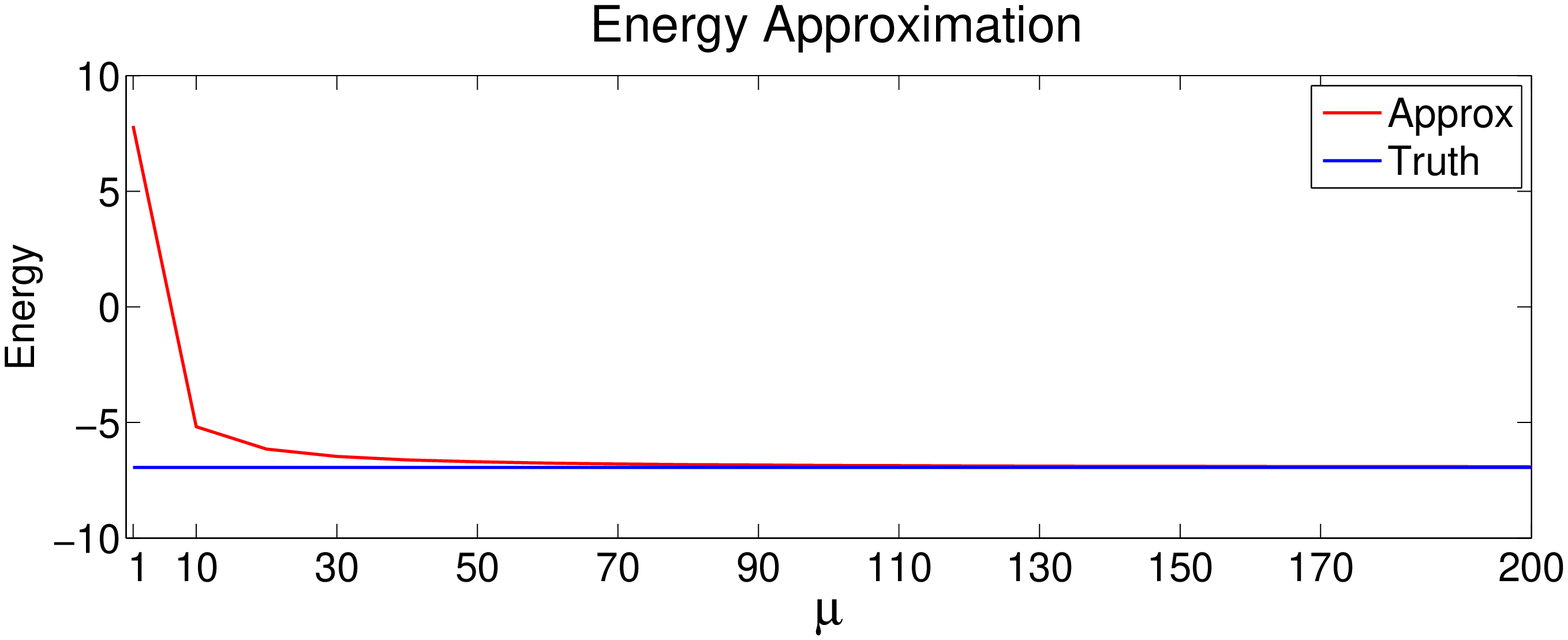}\\
\centering (a)\\
\includegraphics[width=.8\linewidth]{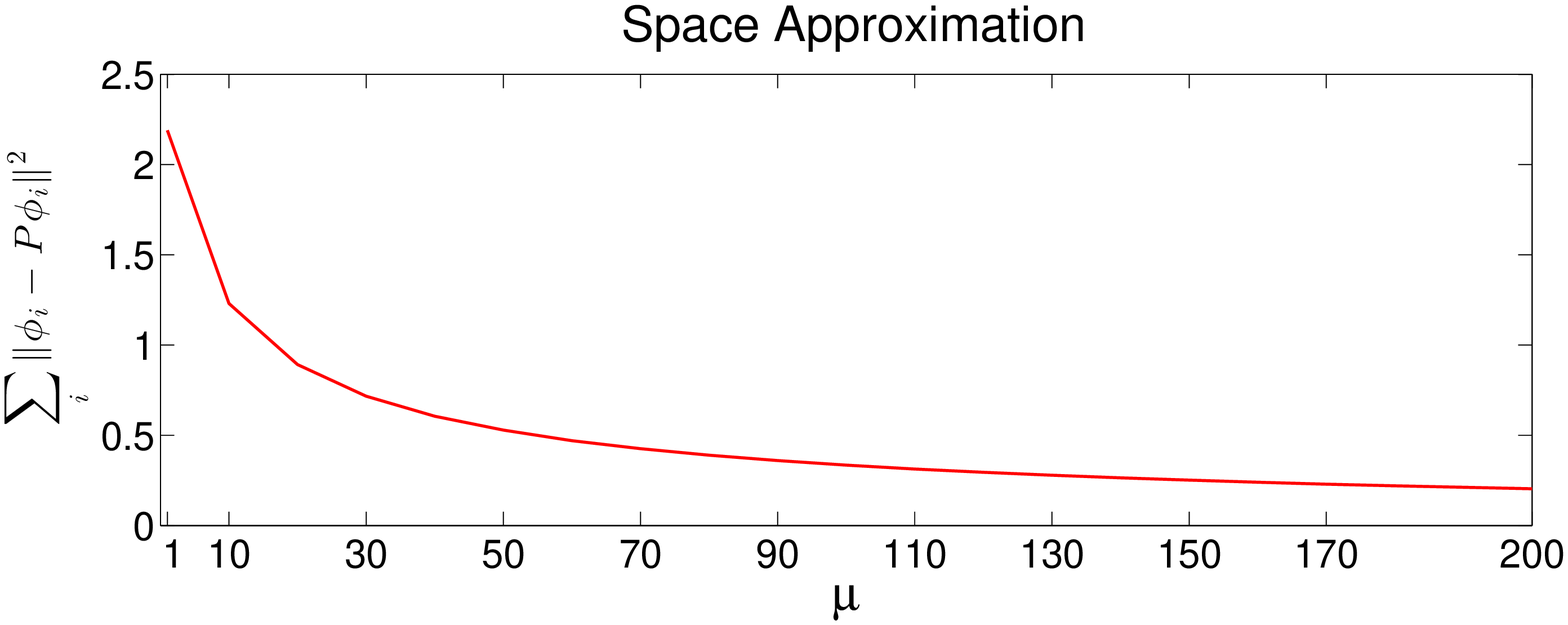}\\
\centering (b)
\caption{(a): Energy approximation as a function of $\mu$. (b): Space
  approximation as a function of $\mu$.}
\label{fig:DensityFunApprox_KP_energy}
\end{figure}

Next we consider the case $N = 15$. The first band of $10$ eigenstates
of $H$ is occupied and the second band of $H$ is ``half-filled''. That
is we have only $5$ electrons occupying the $10$ eigenstates of
comparable eigenvalue of $H$. Hence, the system does not have a gap,
which is indicated by the slow decay of the density matrix shown in
Figure~\ref{fig:KP_N15}(a). Nevertheless, the algorithm with $\mu =
100$ gives a sparse representation of the density matrix, which
captures the feature of the density matrix near the diagonal, as shown in Figure~\ref{fig:KP_N15}(b). To understand better the resulting sparse representation, we diagonal the matrix $P$: 
\begin{equation*}
  P = \sum_{i} f_i \varphi_i \varphi_i^{\TT}. 
\end{equation*}
The eigenvalues $f_i$, known as the occupation number in the physics
literature, are sorted in the decreasing order. The first $40$
occupation numbers are shown in Figure~\ref{fig:KP_N15}(c). We have
$\sum_i f_i = \tr P = 15$, and we see that $\{f_i\}$ exhibits two
groups. The first $10$ occupation numbers are equal to $1$,
corresponding to the fact that the lowest $10$ eigenstates of the
Hamiltonian operator is occupied. Indeed, if we compare the
eigenvalues of the operator $PH$ with the eigenvalues of $H$, as in
Figure~\ref{fig:KP_N15}(d), we see that the first $10$ low-lying
states are well represented in $P$. This is further confirmed by the filtered density matrix $M_1$ given by the first $10$ eigenstates of $P$ as 
\begin{equation*}
  M_1 = \sum_{i=1}^{10}  f_i \varphi_i \varphi_i^{\TT},
\end{equation*}
plotted in Figure~\ref{fig:KP_N15}(e). It is clear that it is very
close to the exact density matrix corresponding to the first $10$
eigenfunctions of $H$, as plotted in
Figure~\ref{fig:DensityFunApprox_KP}(a). The next group of occupation
numbers in Figure~\ref{fig:KP_N15}(c) gets value close to $0.5$. This
indicates that those states are ``half-occupied'', matches very well
with the physical intuition. This is also confirmed by the state
energy shown in Figure~\ref{fig:KP_N15}(d). Note that due to the fact
these states are half filled, the perturbation in the eigenvalue by
the localization is much stronger. The corresponding filtered density
matrix
\begin{equation*}
  M_2 = \sum_{i=11}^{20}  f_i \varphi_i \varphi_i^{\TT},
\end{equation*}
is shown in Figure~\ref{fig:KP_N15}(f). 

For this example, we compare with the results obtained using the
variational principle \eqref{eq:Psi} as in
\cite{OzolinsLaiCaflischOsher:13} shown in
Figure~\ref{fig:CMs_KP_N15}. As the variational principle
\eqref{eq:Psi} is formulated with orbital functions $\Psi$, it does
not allow fractional occupations, in contrast with the one in terms of
the density matrix. Hence, the occupation number is either $1$ or $0$,
which is equivalent to the idempotency condition, as shown in
Figure~\ref{fig:CMs_KP_N15}(b). As a result, even though the states in
the second band have very similar energy, the resulting $\Psi$ are
forced to choose five states over the ten, as can be seen from the
Ritz value plotted in Figure~\ref{fig:CMs_KP_N15}(c). The solution is
quite degenerate in this case. Physically, what happens is that the
five electrons choose $5$ wells out of the ten to sit in (on top of
the state corresponding to the first band already in the well), as
shown from the corresponding density matrix in
Figure~\ref{fig:CMs_KP_N15}(a), or more clearly by the filtered
density matrix in Figure~\ref{fig:CMs_KP_N15}(d) for the five higher
energy states.

\begin{figure}[ht]
\centering
\begin{minipage}{0.49\linewidth}
\includegraphics[width=1\linewidth]{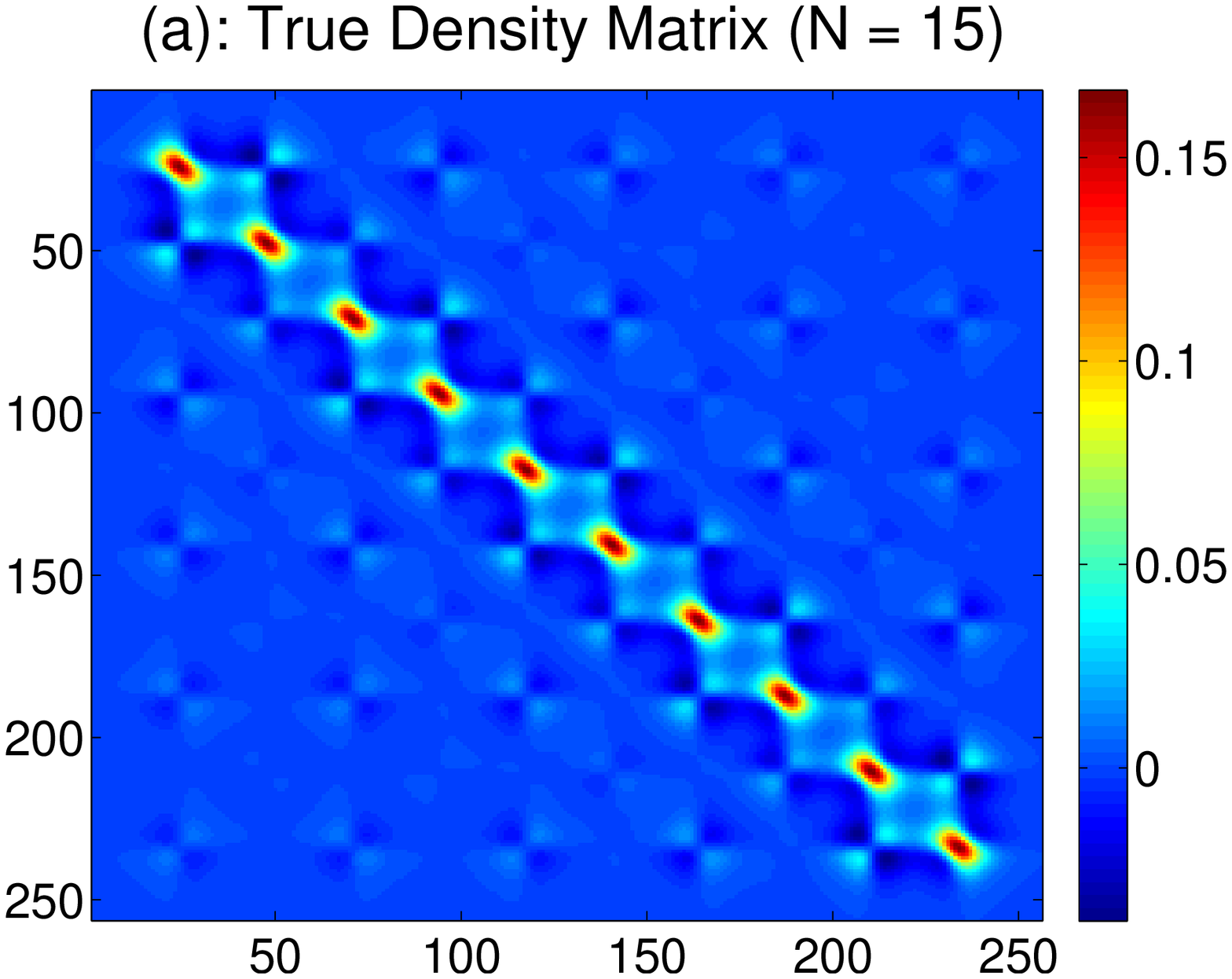}\\
\end{minipage}\hfill
\begin{minipage}{0.49\linewidth}
\includegraphics[width=1\linewidth]{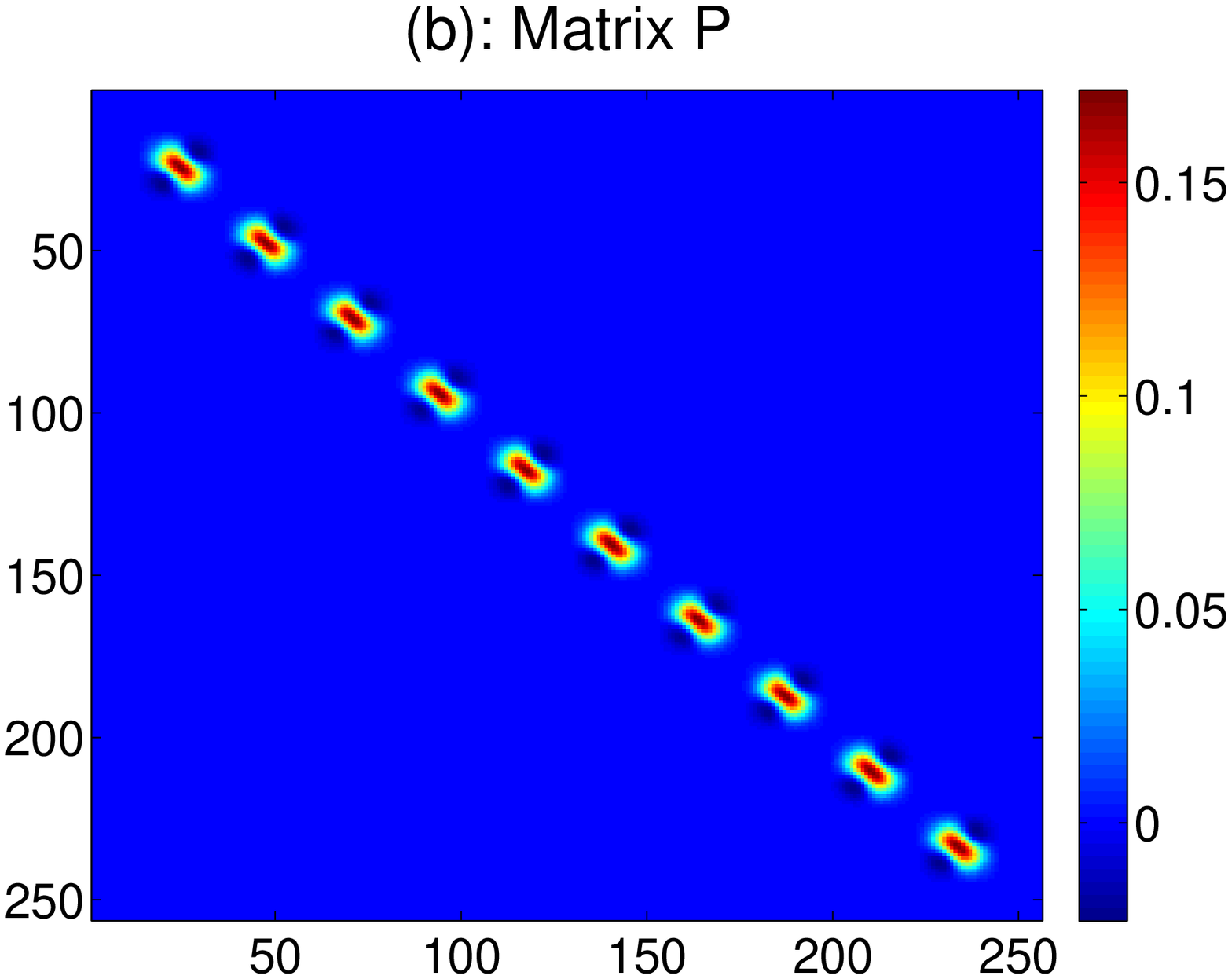}\\
\end{minipage}\hfill\\
\begin{minipage}{0.49\linewidth}
\includegraphics[width=.9\linewidth]{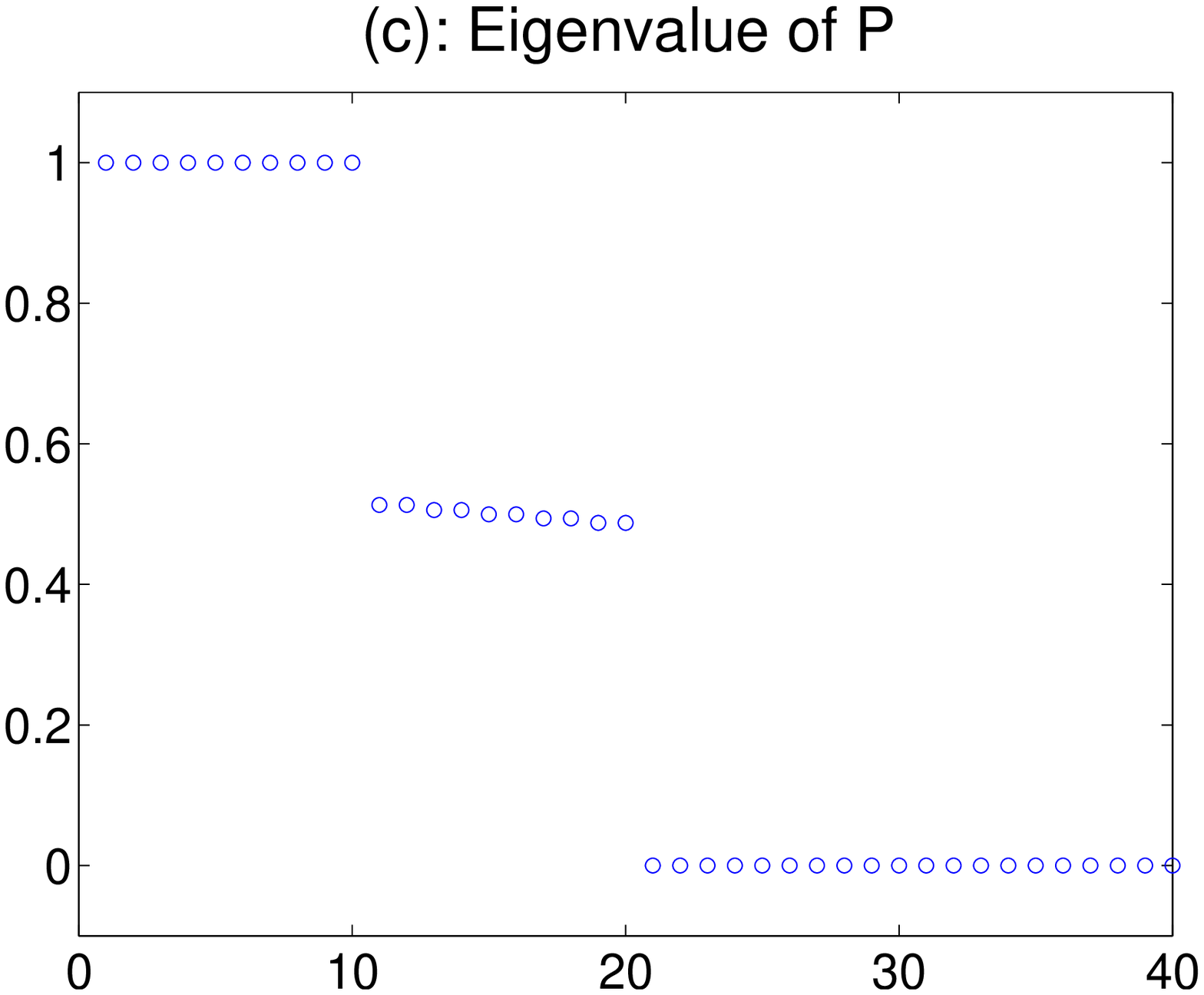}\\
\end{minipage}\hfill
\begin{minipage}{0.49\linewidth}
\includegraphics[width=.9\linewidth]{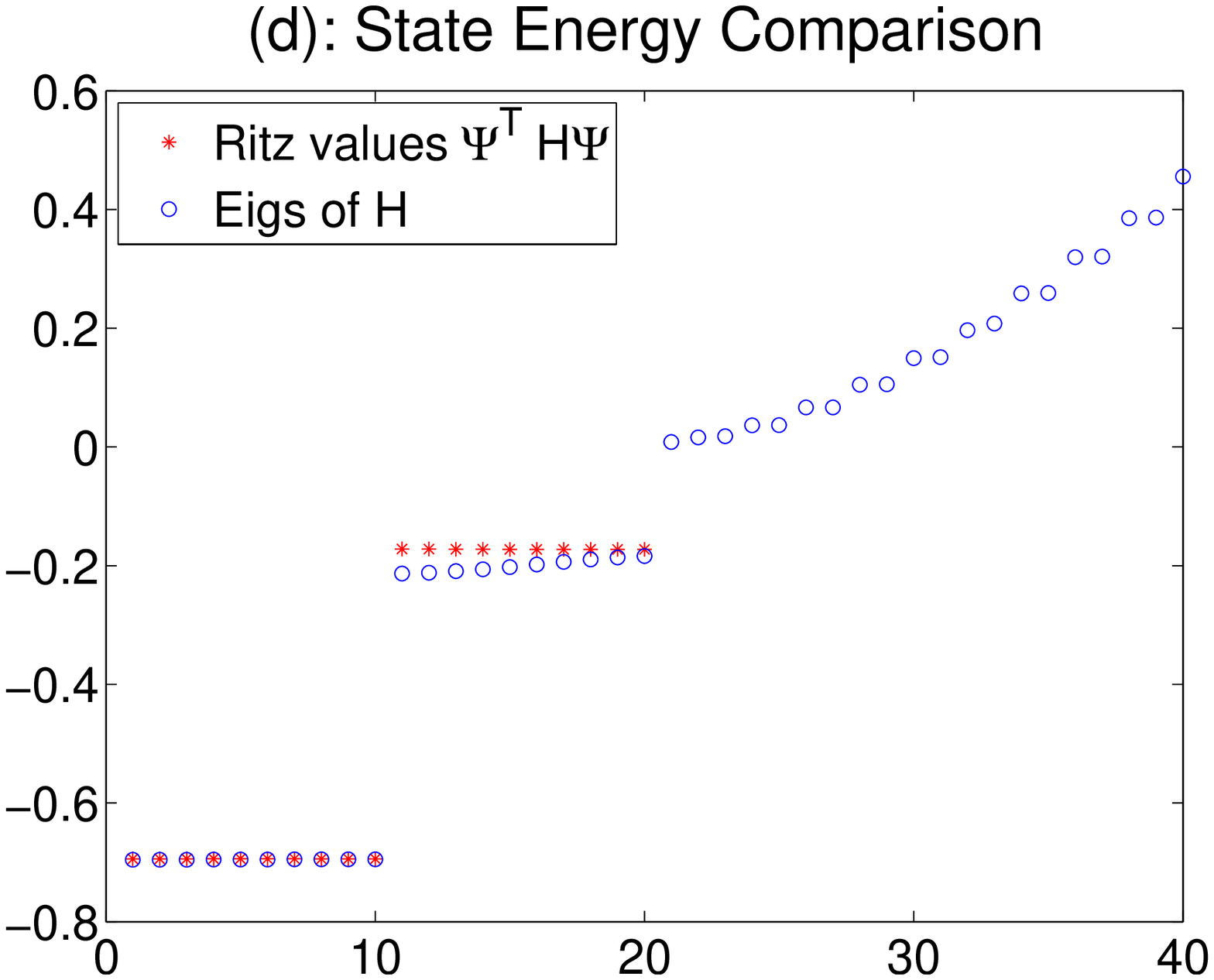}\\
\end{minipage}\hfill\\
\begin{minipage}{0.49\linewidth}
\includegraphics[width=1\linewidth]{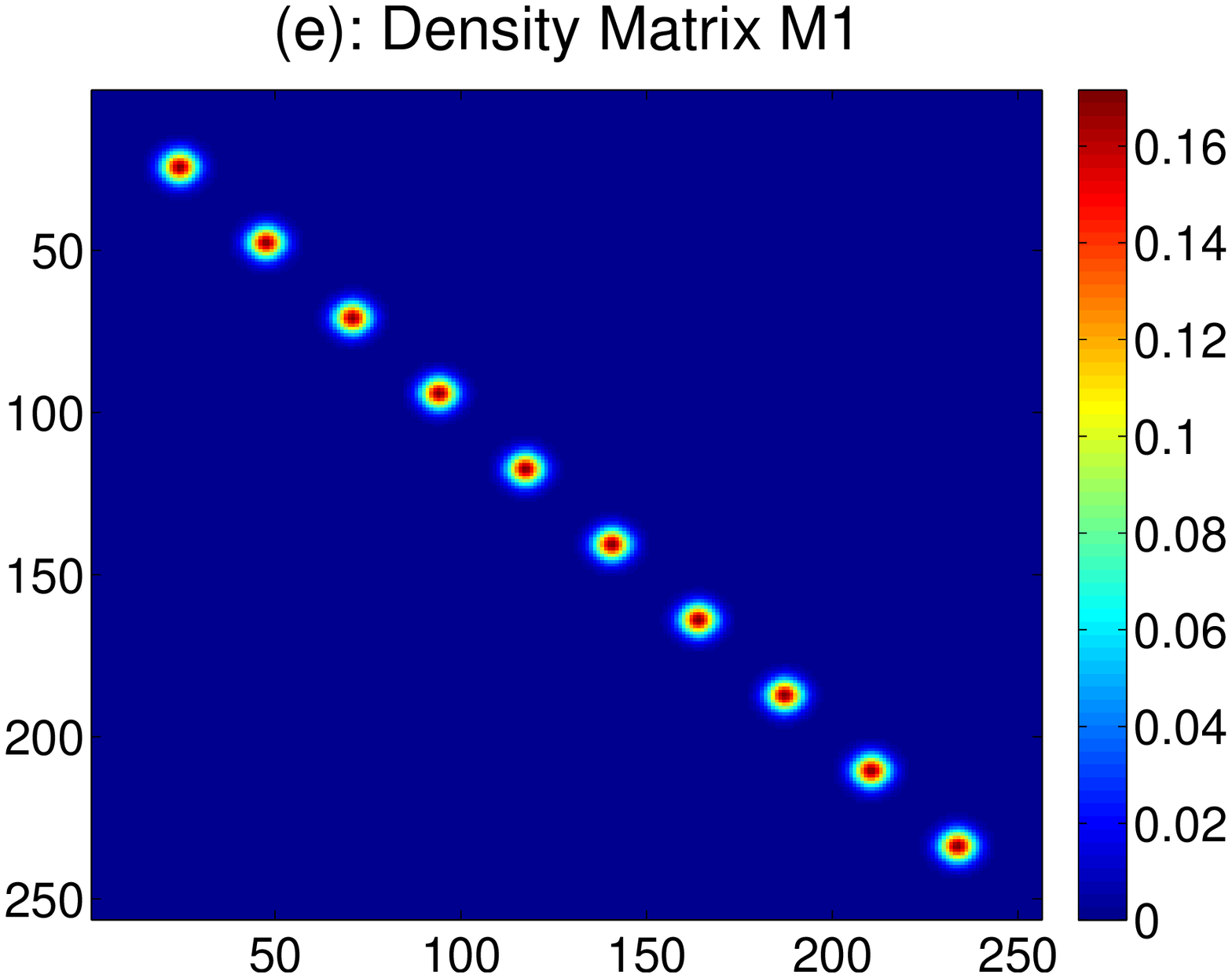}\\
\end{minipage}\hfill
\begin{minipage}{0.49\linewidth}
\includegraphics[width=1\linewidth]{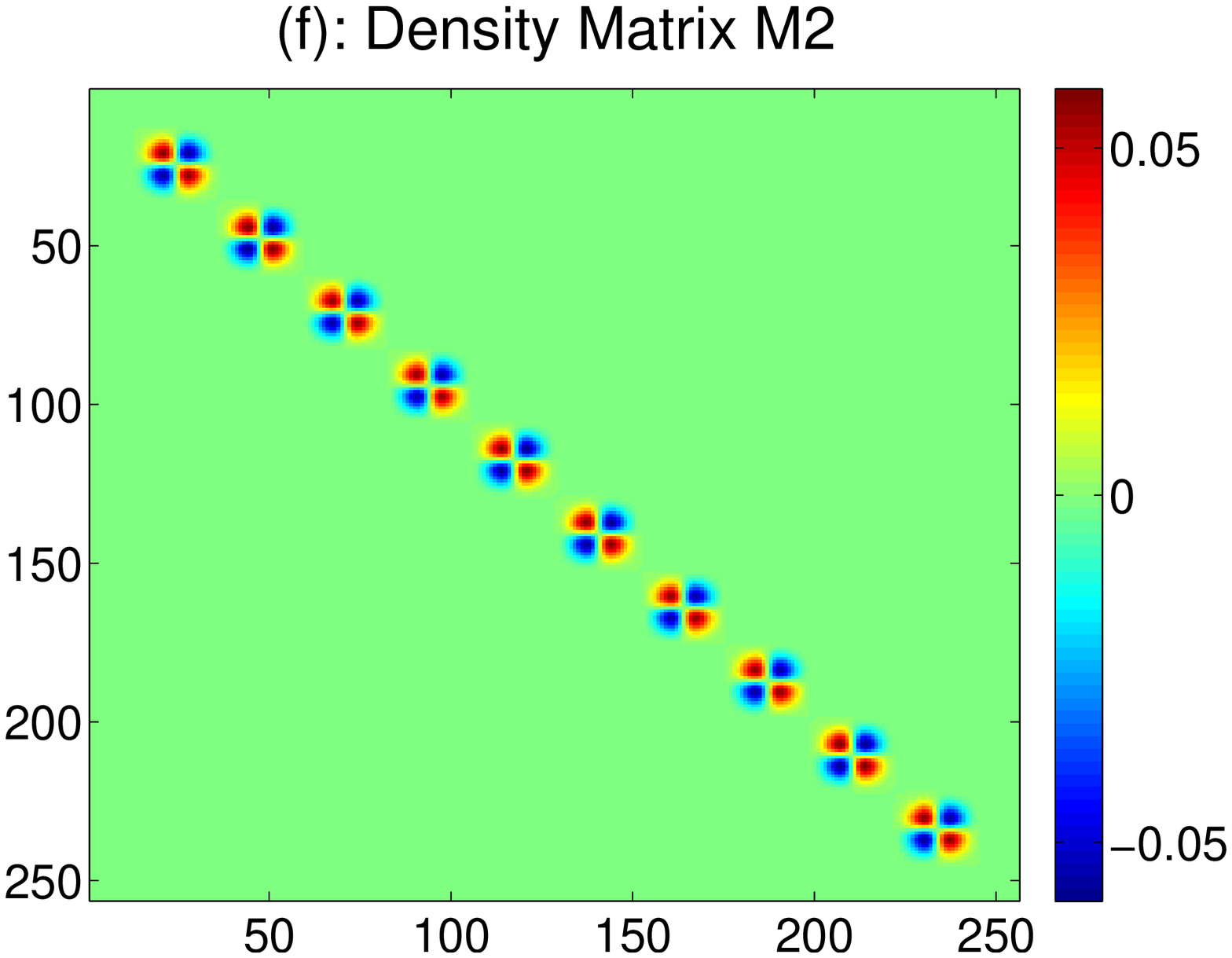}\\
\end{minipage}
\caption{(a): The true density matrix corresponds to the first $15$
  eigenfunctions of $H$. (b): The sparse representation $P$ of the
  density matrix for $\mu = 100$.  (c): The occupation number
  (eigenvalues) of $P$. (d) The first $15$ eigenvalues of $PH$
  compared with the eigenvalues of $H$. (e): The filtered density
  matrix $M_1$ corresponds to the first $10$ eigenstates of $P$.  (f)
  The filtered density matrix $M_2$ corresponds to the next $10$
  eigenstates of $P$.\label{fig:KP_N15}}
\end{figure}

\begin{figure}[ht]
\centering
\begin{minipage}{0.49\linewidth}
\includegraphics[width=1\linewidth]{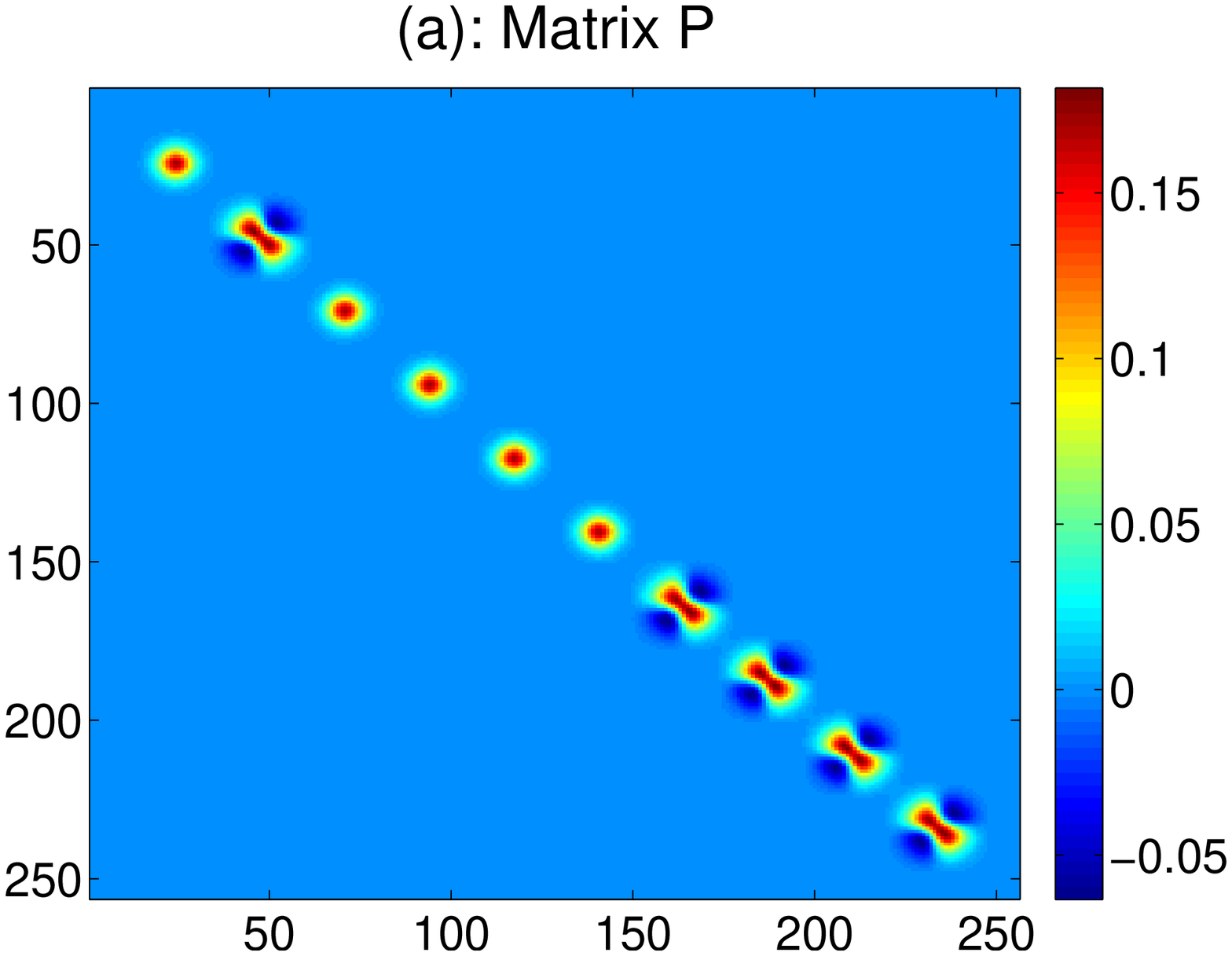}\\
\end{minipage}\hfill
\begin{minipage}{0.49\linewidth}
\includegraphics[width=.9\linewidth]{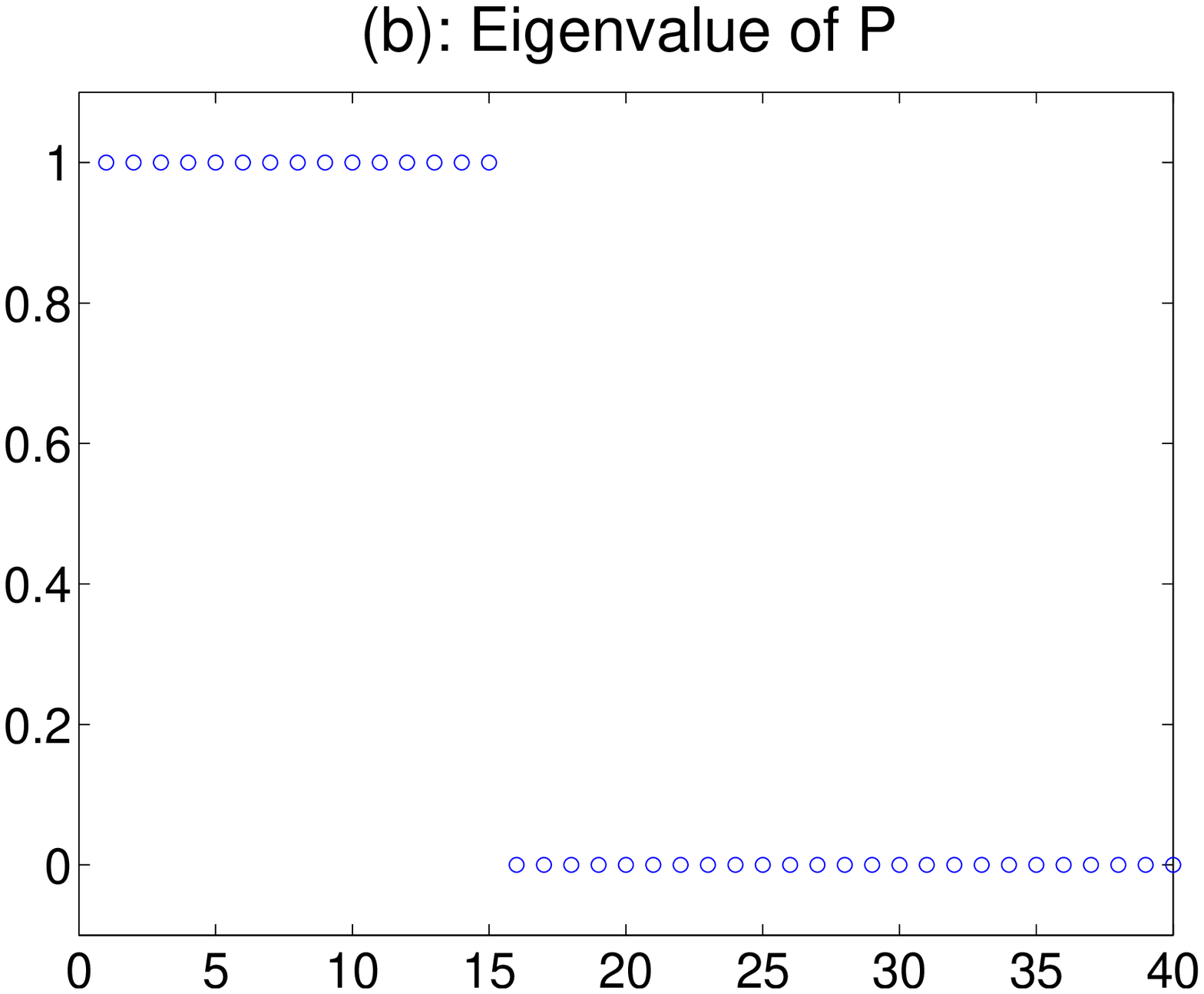}\\
\end{minipage}\hfill\\
\begin{minipage}{0.49\linewidth}
\includegraphics[width=.9\linewidth]{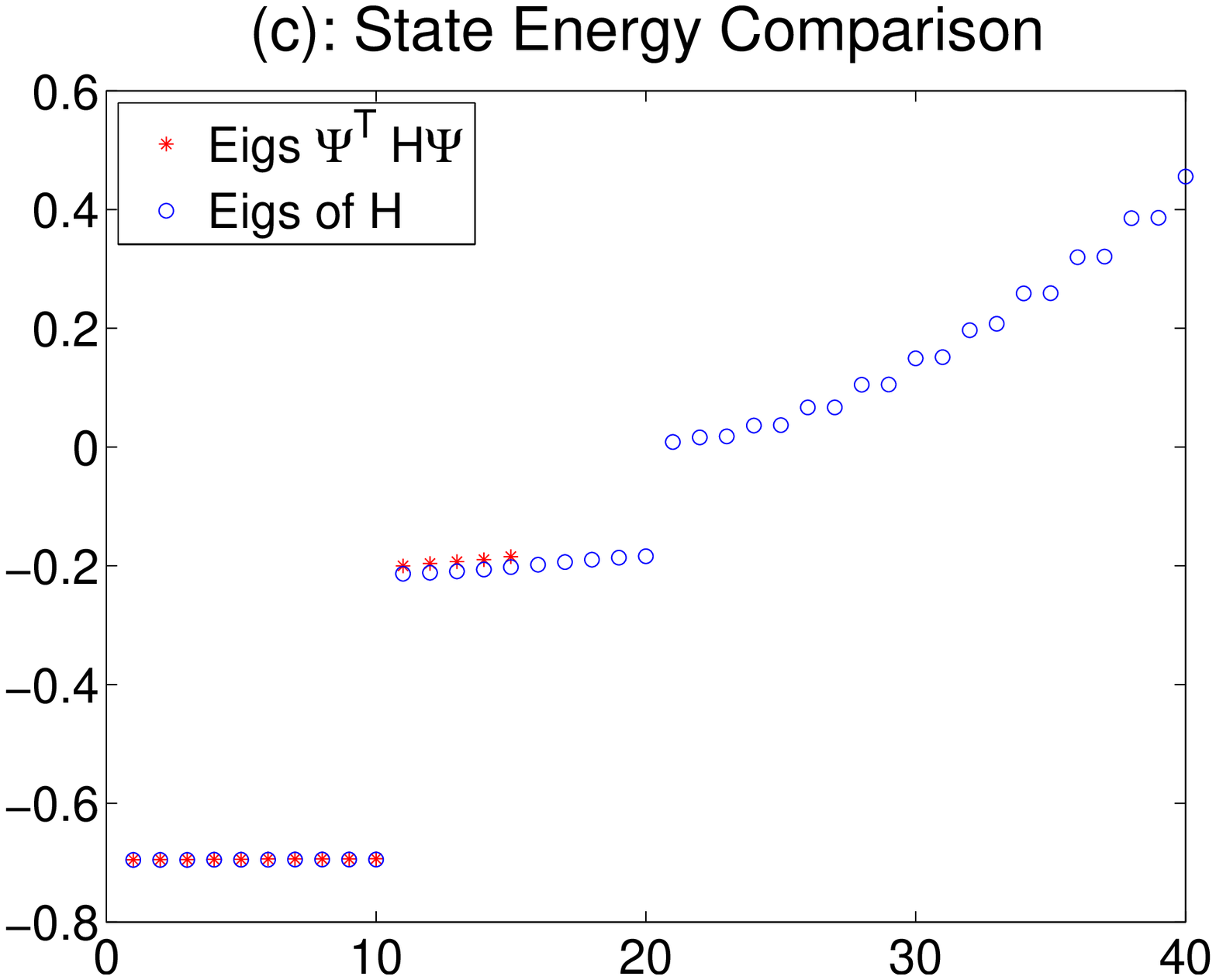}\\
\end{minipage}\hfill
\begin{minipage}{0.49\linewidth}
\includegraphics[width=1\linewidth]{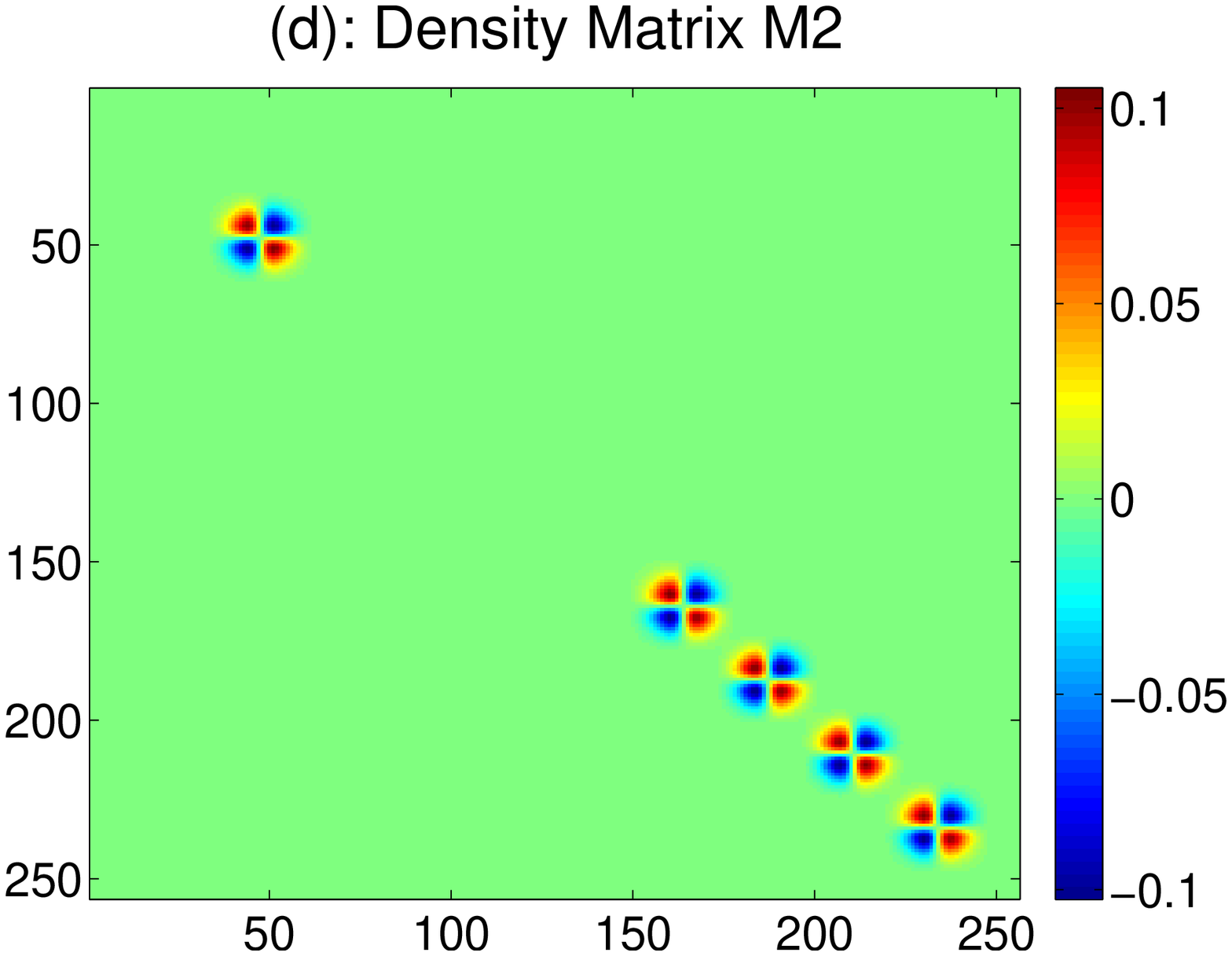}\\
\end{minipage}
\caption{Results obtained by the first 15 Compressed modes $\Psi =
  \{\psi_i\}_{i=1}^{15}$ for $\mu = 100$. (a): The density
  representation $P$ given by $P = \Psi^T\Psi $.  (b): The occupation
  number (eigenvalues) of $P$. (d) The first $15$ eigenvalues of
  $\Psi^T H \Psi$ compared with the eigenvalues of $H$. (d) The
  filtered density matrix $M_2$ corresponds to the $5$ states in the
  second band.
  \label{fig:CMs_KP_N15}}
\end{figure}

Finally, the $N = 20$ case corresponds to the physical situation that
the first two bands are all occupied. Note that as the band gap
between the second band from the rest of the spectrum is smaller than
the gap between the first two bands, the density matrix, while still
exponentially localized, has a slower off diagonal decay rate. The
exact density matrix corresponds to the first $20$ eigenfunctions of
$H$ is shown in Figure~\ref{fig:KP_N20}(a), and the localized
representation with $\mu = 100$ is given in
Figure~\ref{fig:KP_N20}(b). The occupation number is plotted in
Figure~\ref{fig:KP_N20}(c), indicates that the first $20$ states are
fully occupied, while the rest of the states are empty. This is
further confirmed by comparison of the eigenvalues given by $HP$ and
$H$, shown in Figure~\ref{fig:KP_N20}(d). In this case, we see that
physically, each well contains two states. Hence, if we look at the
electron density, which is diagonal of the density matrix, we see a
double peak in each well. Using the projection of delta functions, we
see that the sparse representation of the density matrix $P$
automatically locate the two localized orbitals centered at the two
peaks, as shown in Figure~\ref{fig:KP_N20}(e).

\begin{figure}[ht]
\centering
\begin{minipage}{0.49\linewidth}
\includegraphics[width=1\linewidth]{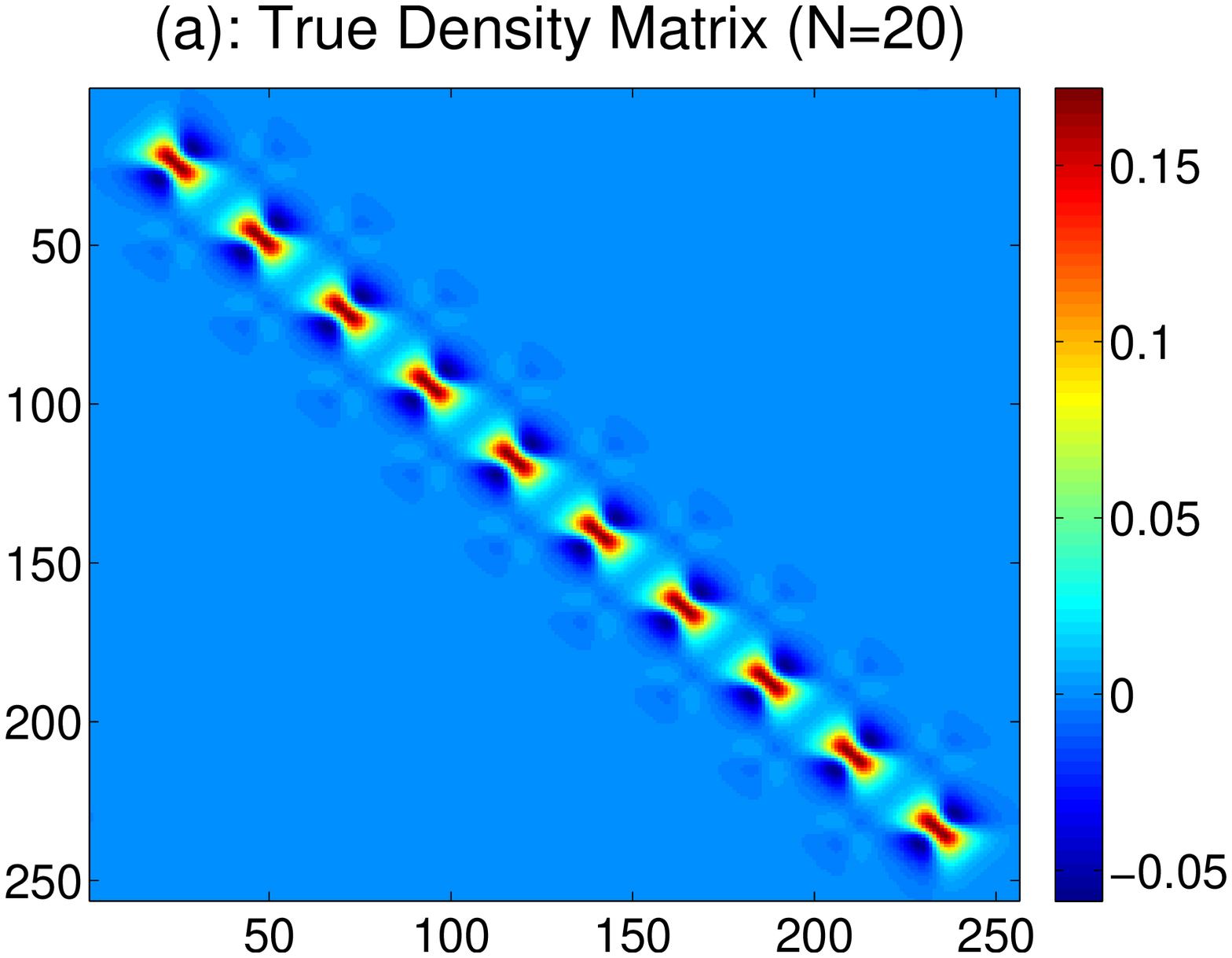}\\
\end{minipage}\hfill
\begin{minipage}{0.49\linewidth}
\includegraphics[width=1\linewidth]{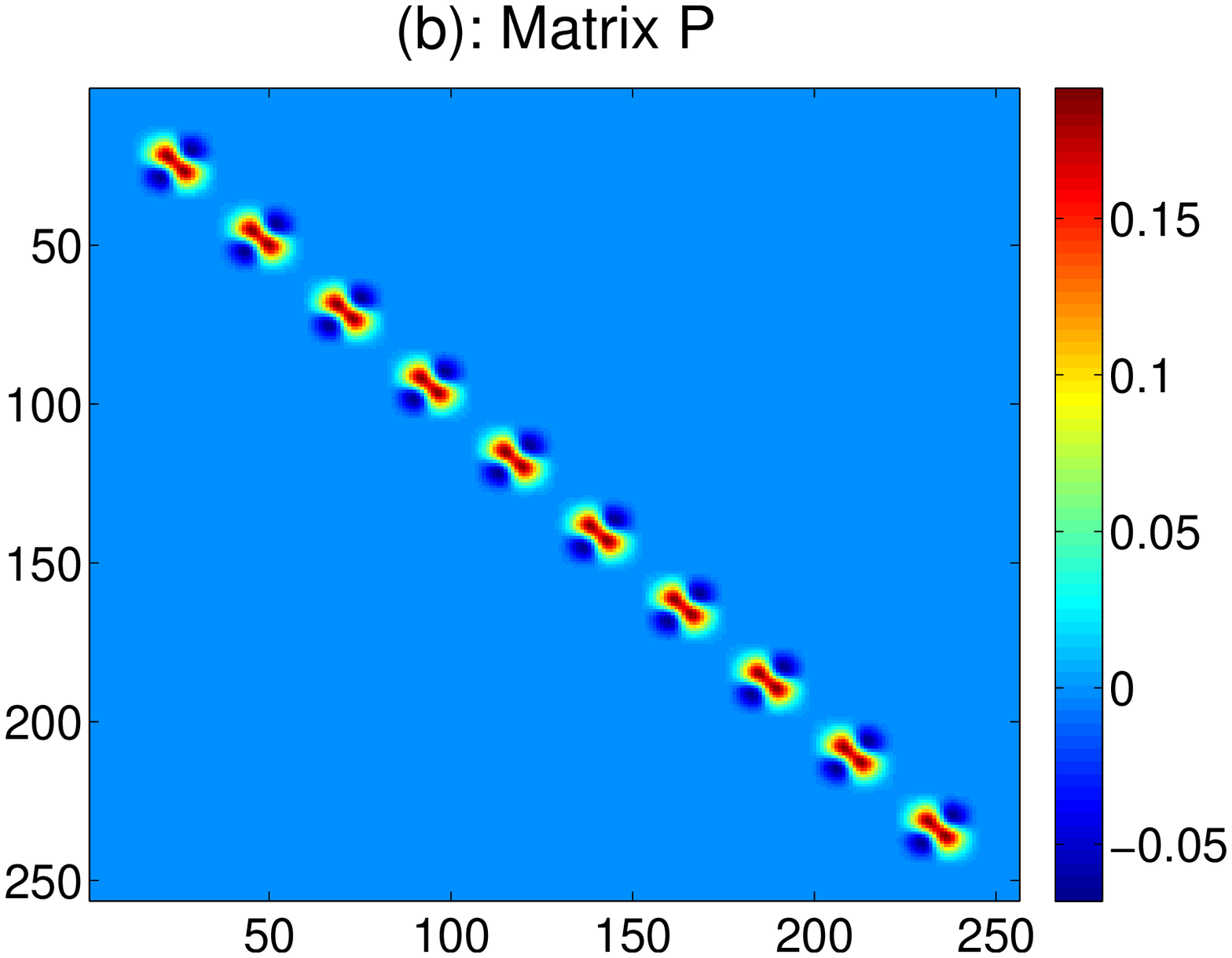}\\
\end{minipage}\hfill\\
\begin{minipage}{0.49\linewidth}
\includegraphics[width=.9\linewidth]{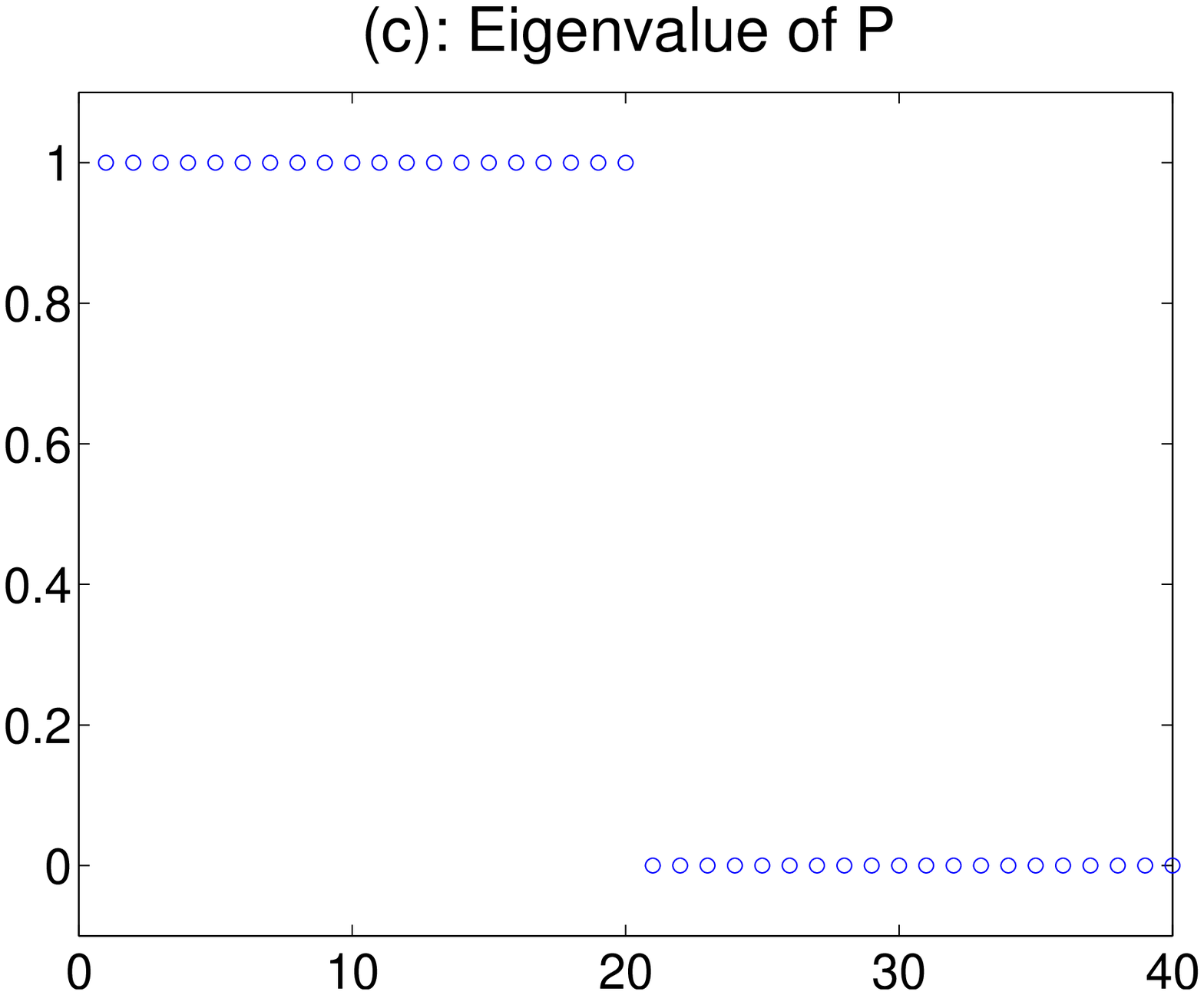}\\
\end{minipage}\hfill
\begin{minipage}{0.49\linewidth}
\includegraphics[width=.9\linewidth]{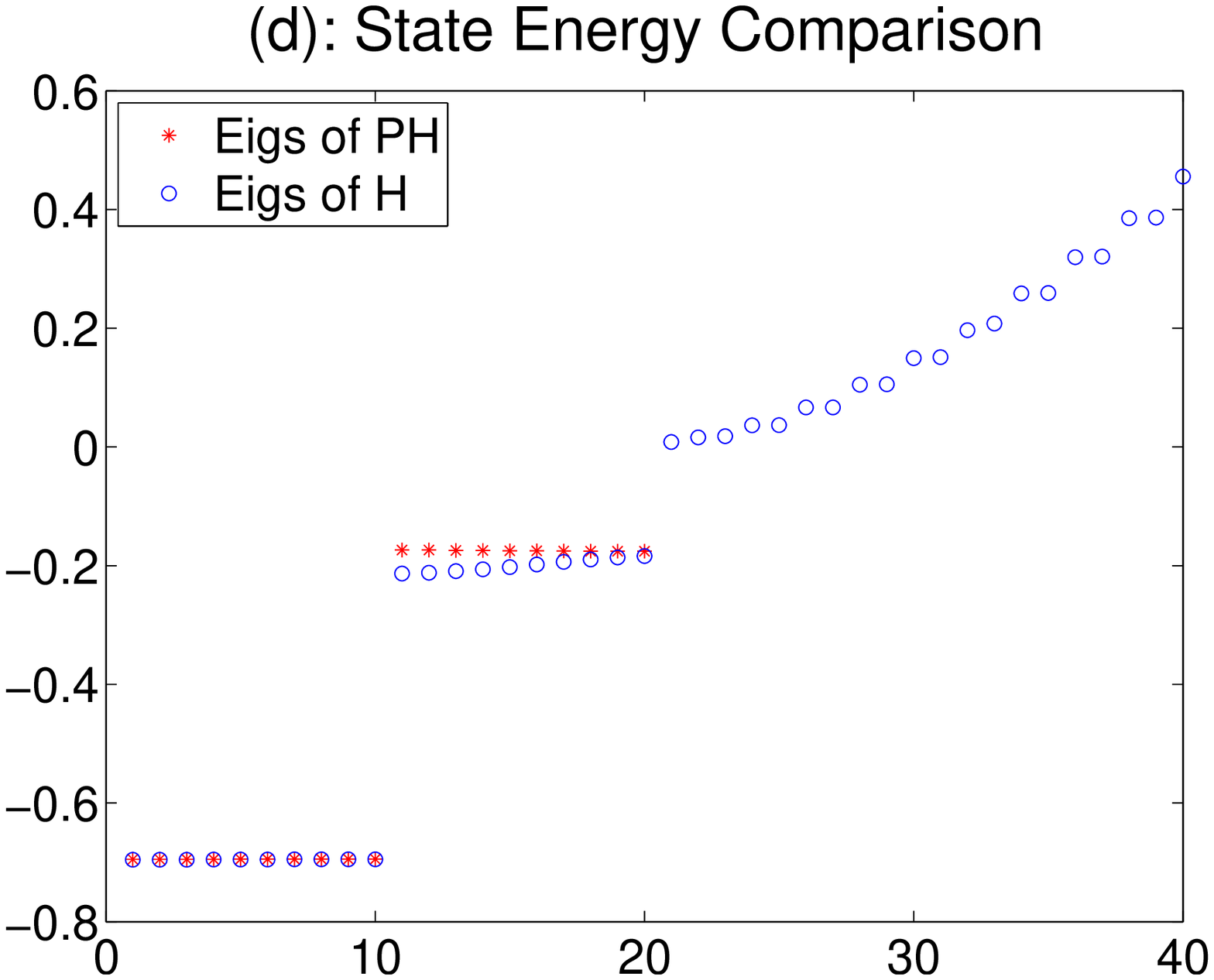}\\
\end{minipage}\hfill\\
\centering
\includegraphics[width=1\linewidth]{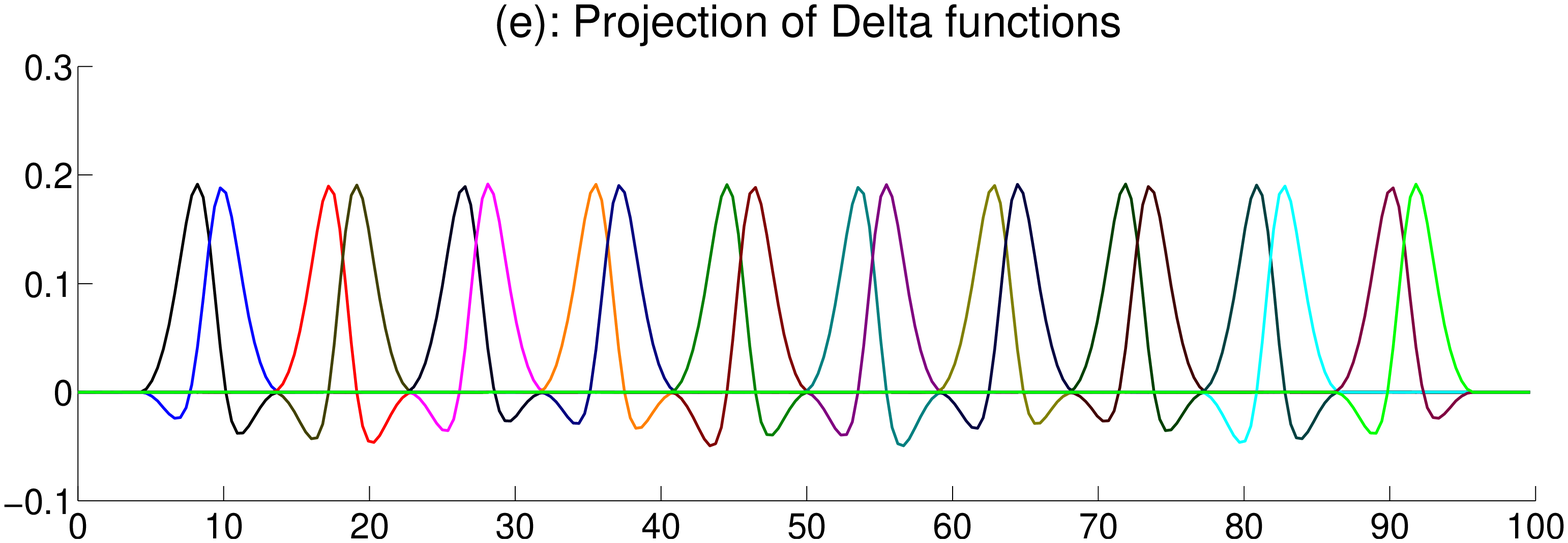}\\
\caption{(a): The true density matrix
  corresponds to the first $20$ eigenfunctions of $H$. (b): The sparse
  representation $P$ of the density matrix for $\mu = 100$.  (c): The
  occupation number (eigenvalues) of $P$. (d) The first $20$
  eigenvalues of $PH$ compared with the eigenvalues of $H$. (e)
  Projection of Delta function $\delta(x - x_i)$.}
\label{fig:KP_N20}
\end{figure}

\subsection{An example of non-unique minimizers}

Let us revisit the Example~$2$ in Section~\ref{sec:formulation} for
which the minimizers to the variational problem is
non-unique. Theorem~\ref{thm:conv} guarantees that the algorithm will
converge to some minimizer starting from any initial condition. 

It is easy to check that in this case
\begin{equation}
  P^{\ast} = Q^{\ast} = R^{\ast} = 
  \begin{pmatrix}
    1 & 0 & 0 \\
    0 & 0 & 0 \\
    0 & 0 & 0
  \end{pmatrix}, 
  \quad 
   b^{\ast} = \begin{pmatrix} 
     1 & 0 & 0 \\ 
     0 & 1 & -1 \\
     0 & -1 & 1 
   \end{pmatrix}, \quad d^{\ast} = \begin{pmatrix}
     0 & 0 & 0 \\
     0 & -1 & -1 \\
     0 & -1 & -1
  \end{pmatrix}
\end{equation}
is a fixed point of the algorithm. In Figure~\ref{fig:DecayDist}, we
plot the sequence $\Big\{ \lambda \|b^{k} - b^{\ast}\|^2 + r \|d^{k} -
d^{\ast}\|^2 + \lambda \| Q^{k} - Q^{\ast} \|^2 + r \| R^{k} -
R^{\ast} \|^2 \Big\}_k$ for a randomly chosen initial data. We see
that the distance does not converge to $0$ as the algorithm converges
to another minimizer of the variational problem. Nonetheless, as will
be shown in the proof of Theorem~\ref{thm:conv} in
Section~\ref{sec:proof}, the sequence is monotonically non-increasing.

\begin{figure}[ht]
\centering
\includegraphics[width=.8\linewidth]{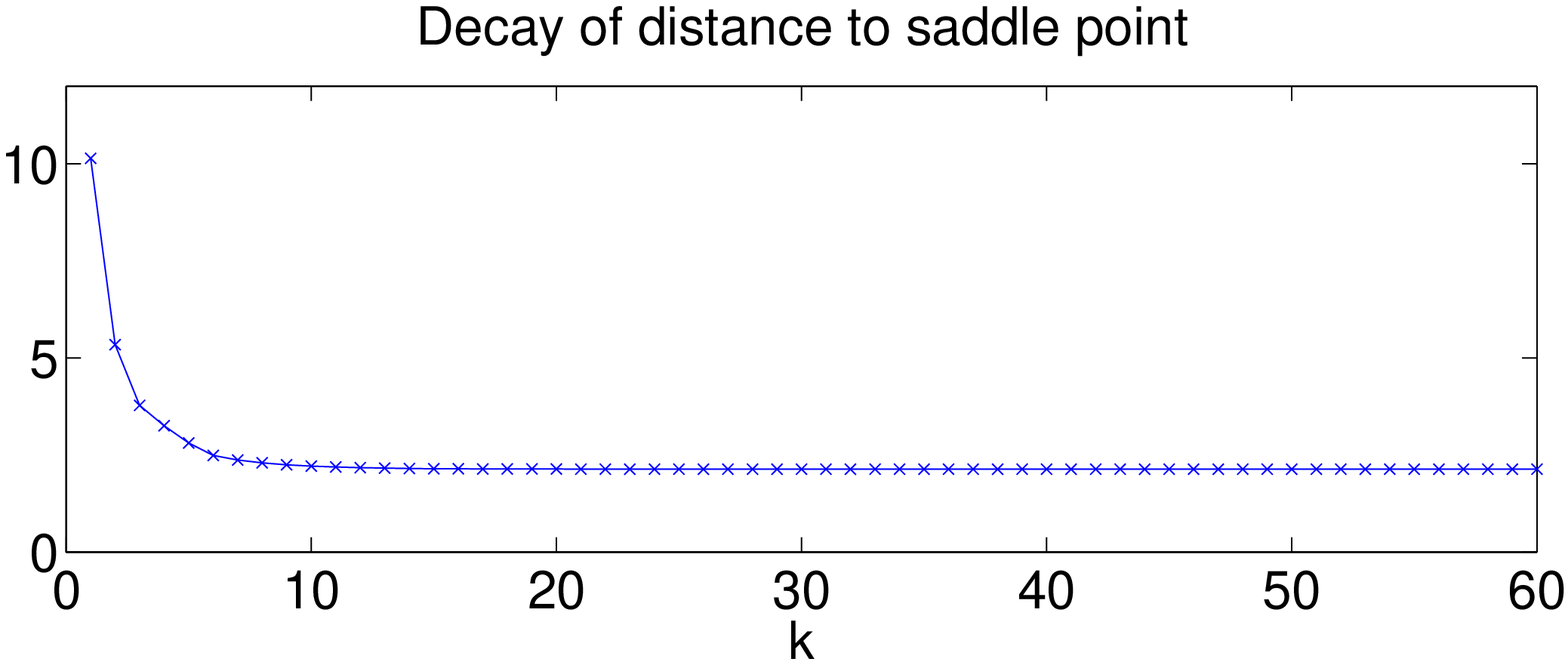}
\caption{$\lambda \|b^{k} - b^{\ast}\|^2 + r \|d^{k} -
d^{\ast}\|^2 + \lambda \| Q^{k} - Q^{\ast} \|^2 + r \| R^{k} -
R^{\ast} \|^2$ as a function of $k$  for Algorithm~\ref{alg:CM_P}.}\label{fig:DecayDist}
\end{figure}

\section{Convergence of Algorithm~\ref{alg:CM_P}}
\label{sec:proof}

For ease of notation, we will prove the convergence of the algorithm for the following slightly generalized variational problem. 
\begin{equation}\label{eq:three}
  \begin{aligned}
    & \min_{P,Q,R} f(P) + g(Q) + h(R) \\
    & \text{s.t.} \; P = Q, P = R
  \end{aligned}
\end{equation}
where $f$, $g$, and $h$ are proper convex functionals, but not
necessarily strictly convex. In particular, we will get \eqref{eq:P_split} if we set 
\begin{align*}
  & f(P) = 
  \begin{cases} 
    \tr(H P), & \text{if } \tr P = N; \\
    + \infty, & \text{otherwise}, 
  \end{cases} \\
  & g(Q) = \norm{Q}_1 / \mu \\
  & h(R) = 
  \begin{cases}
    0, & \text{if } R = R^{\TT}, \text{and } 0 \preceq R \preceq I; \\
    + \infty, & \text{otherwise}. 
  \end{cases}
\end{align*}

The corresponding algorithm for \eqref{eq:three} is given by 
\begin{algorithm}\label{alg:split2}
Initialize $P^0 = Q^0 = R^0, b^0 = d^0 = 0$

\While{``not converge"}{
\begin{enumerate}
\item $\displaystyle P^k = \arg\min_{P} f(P) + \frac{\lambda}{2} \norm{P - Q^{k-1} + b^{k-1}}^2 + \frac{r}{2} \norm{P - R^{k-1} + d^{k-1}}^2$, 
\item $\displaystyle Q^k = \arg\min_{Q} g(Q) + \frac{\lambda}{2} \norm{P^k - Q + b^{k-1}}^2 $.
\item $R^k = \displaystyle\arg\min_{R}  h(R) +  \frac{r}{2}\norm{P^{k} - R + d^{k-1}}^2$. 
\item $b^{k} = b^{k-1} +  P^k - Q^k$.
\item $d^{k} = d^{k-1} +  P^k - R^k$.
\end{enumerate}
}
\end{algorithm}

We define an augmented Lagrangian 
\begin{multline}\label{eq:Lag}
  \mc{L}(P,Q,R;b,d) = f(P) + g(Q) + h(R) + \frac{\lambda}{2}\|P - Q\|^2 + \lambda\langle P - Q, b\rangle \\
  + \frac{r}{2}\|P - R\|^2 + r\langle P - R, d\rangle
\end{multline}

\begin{defn}
  We call $(P^{\ast}, Q^{\ast}, R^{\ast}; b^{\ast}, d^{\ast})$ a
  \emph{saddle point} of the Lagrangian \eqref{eq:Lag}, if
  \begin{equation} \label{eq:saddle}
    \mc{L}(P^{\ast}, Q^{\ast}, R^{\ast}; b, d) \leq \mc{L}(P^{\ast}, Q^{\ast}, R^{\ast}; b^{\ast}, d^{\ast}) \leq \mc{L}(P, Q, R; b^{\ast}, d^{\ast})
  \end{equation}
  for any $(P, Q, R; b, d)  \in \RR^{n\times n} \times \RR^{n\times n} \times \RR^{n\times n}\times \RR^{n\times n}\times \RR^{n\times n}$. %\textbf{need to specify some space?}
\end{defn}

\begin{lemma} $(\Ps, \Qs, \Rs)$ is a solution of the optimization problem \eqref{eq:three} if any only if there exist $\bs, \ds \in \RR^{n\times n}$ such that $(\Ps, \Qs, \Rs;\bs, \ds)$ is a saddle point satisfying \eqref{eq:saddle}
\end{lemma}
\begin{proof} Given a saddle point $(\Ps, \Qs, \Rs;\bs, \ds)$ satisfying \eqref{eq:saddle}, it is clear that the first inequality in \eqref{eq:saddle} implies $\Ps = \Qs = \Rs$. Substitute $P = Q = R$ in the second inequality of \eqref{eq:saddle}, we can immediately have $(\Ps, \Qs, \Rs)$ is a minimizer \eqref{eq:three}. 

On the other hand, suppose $(\Ps, \Qs, \Rs)$ is a solution of \eqref{eq:three}. The first inequality in \eqref{eq:saddle} holds since $\Ps = \Qs = \Rs$. Moreover, there exist $\bs, \ds$ such that 
\begin{equation*}
-\lambda \bs - r \ds \in \partial f(\Ps), \qquad \lambda \bs \in \partial g(\Qs), \qquad r \ds \in \partial h(\Rs)
\end{equation*}
which suggests, for any $P, Q, R \in \RR^{n\times n}$
\begin{align*}
f(\Ps)  &\leq f(P) + \lambda \langle \bs, P - \Ps \rangle + r \langle \ds, P - \Ps \rangle \nonumber \\
g(\Qs) &\leq g(Q) - \lambda \langle \bs, Q - \Qs \rangle  \nonumber \\
h(\Rs) &\leq h(R)  -  r \langle \ds, R - \Rs \rangle \nonumber 
\end{align*}
The summation of the above three inequalities yield the second inequality in \eqref{eq:saddle}.
\end{proof}
% \begin{equation}
%   f(P^k) + g(Q^k) + h(R^k) + \average{b, P^k - Q^k} + \average{d, P^k - R^k} 
% \geq f(P) + g(Q) + h(R).
% \end{equation}

% \begin{align}
%   & f(P) - f(P^k) + \lambda \average{P^k - Q^{k-1} + b^{k-1}, P - P^k} + 
%   r \average{P^k - R^{k-1} + d^{k-1}, P - P^k} \geq 0  \\
%   & g(Q) - g(Q^k) + \lambda \average{P^k - Q^k + b^{k-1}, Q^k - Q} \geq 0 \\
%   & h(R) - h(R^k) + r \average{P^k - R^k + d^{k-1}, R^k - R} \geq 0 
% \end{align}

\begin{theorem}\label{thm:three}
  The sequence $\Big\{(P^k, Q^k, R^k)\Big\}_k$ generated by
  Algorithm~\ref{alg:split2} from any starting point converges to a
  minimum of the variational problem \eqref{eq:three}.
\end{theorem}

\begin{remark}
  We remind the readers that the minimizers of the variational
  principle \eqref{eq:three} might not be unique. In the non-unique
  case, the above theorem states that any initial condition will
  converge to some minimizer, while different initial condition might
  give different minimizers.
\end{remark}

\begin{proof}
Let $(P^{\ast}, Q^{\ast}, R^{\ast})$ be an optimal solution of \eqref{eq:three}. 
We introduce the short hand notations
\begin{equation}
  \begin{split}
    & \wb{P}^k = P^k - \Ps, \quad \wb{Q}^k = Q^k - \Qs, \quad
    \text{and} \quad \wb{R}^k = R^k - \Rs. \\
    & \wb{b}^k = b^k - \bs, \quad \wb{d}^k = d^k - \ds.
  \end{split}
\end{equation}
From Step $4$ and $5$ in the algorithm, we get 
\begin{equation}
\wb{b}^k = \wb{b}^{k-1} + \wb{P}^{k} - \wb{Q}^{k}, \quad \text{and} \quad 
\wb{d}^k = \wb{d}^{k-1} + \wb{P}^{k} - \wb{R}^{k},
\end{equation}
and hence
\begin{equation} \label{eqn:bderror}
  \begin{split}
    & \|\wb{b}^{k-1}\|^2  - \|\wb{b}^{k}\|^2  = -2\langle \wb{b}^{k-1}, \wb{P}^{k} - \wb{Q}^{k} \rangle - \|\wb{P}^{k} - \wb{Q}^{k}\|^2  \\
    & \norm{\wb{d}^{k-1}}^2 - \|\wb{d}^{k}\|^2 = - 2\langle
    \wb{d}^{k-1}, \wb{P}^{k} - \wb{R}^{k} \rangle - \|\wb{P}^{k}
    - \wb{R}^{k}\|^2
  \end{split}
\end{equation}

Note that by optimality
\begin{align}
 \Ps = \arg\min_{P\in \mathcal{C}_P} \mathcal{L}(P,\Qs,\Rs;\bs,\ds) \\
 \Qs = \arg\min_{Q\in \mathcal{C}_Q} \mathcal{L}(\Ps,Q,\Rs;\bs,\ds) \\
 \Rs = \arg\min_{R\in \mathcal{C}_R} \mathcal{L}(\Ps,\Qs,R;\bs,\ds) 
\end{align}
Hence, for any $P,  Q, R\in \RR^{n\times n}$, we have
\begin{align}
& f(P) - f(\Ps) + \lambda\langle \Ps - \Qs + \bs, P - \Ps \rangle + r \langle \Ps - \Rs + \ds, P - \Ps \rangle \geq 0 \label{eqn:fPs}\\
& g(Q) - g(\Qs) + \lambda\langle \Qs - \Ps - \bs, Q - \Qs \rangle \geq 0  \label{eqn:gQs} \\
& h(R) - h(\Rs) +   r \langle \Rs - \Ps - \ds, R - \Rs \rangle \geq 0  \label{eqn:hRs}
\end{align}

According to the construction of $\{P^k, Q^k, R^k\}$, for any $P, Q, R  \in \RR^{n\times n}$, we have
\begin{align}
  & 
  \begin{aligned}
    f(P) - f(P^{k}) & + \lambda\langle P^{k} - Q^{k-1} + b^{k-1}, P - P^{k} \rangle \\
    & + r \langle P^{k} - R^{k} + d^{k-1}, P - P^{k} \rangle \geq 0
  \end{aligned} \label{eqn:fPk}\\
  & g(Q) - g(Q^{k}) + \lambda\langle Q^{k}- P^{k} - b^{k-1}, Q - Q^{k} \rangle \geq 0  \label{eqn:gQk} \\
  & h(R) - h(R^{k}) + r \langle R^{k} - P^{k} - d^{k-1}, R - R^{k}
  \rangle \geq 0 \label{eqn:hRk}
\end{align}

Let $P = P^{k}$ in \eqref{eqn:fPs} and $P = \Ps$ in \eqref{eqn:fPk}, their summation yields 
\begin{equation}\label{eqn:fPsk}
  \lambda\langle -\wb{P}^{k} + \wb{Q}^{k-1} - \wb{b}^{k-1}, \wb{P}^{k} \rangle + r \langle -\wb{P}^{k} + \wb{R}^{k-1} - \wb{d}^{k-1}, \wb{P}^{k} \rangle \geq 0
\end{equation}
Similarly, let $Q = Q^{k}$ in \eqref{eqn:gQs} and $Q = \Qs$ in \eqref{eqn:gQk}, and let $R = R^{k}$ in \eqref{eqn:hRs} and $R = \Rs$ in \eqref{eqn:hRk}, we obtain 
\begin{align}
  & \lambda\langle -\wb{Q}^{k} + \wb{P}^{k} + \wb{b}^{k-1},
  \wb{Q}^{k} \rangle \geq 0 \label{eqn:gQsk} \\
  & r \langle -\wb{R}^{k} + \wb{P}^{k} + \wb{d}^{k-1}, \wb{R}^{k}
  \rangle \geq 0 \label{eqn:hRsk}
\end{align}

The summation of \eqref{eqn:fPsk}, \eqref{eqn:gQsk}, and
\eqref{eqn:hRsk} yields
\begin{multline}
  \lambda\langle - \wb{b}^{k-1}, \wb{P}^{k} \rangle + \lambda\langle
  \wb{Q}^{k-1} -\wb{P}^{k} , \wb{P}^{k} \rangle + r \langle -
  \wb{d}^{k-1}, \wb{P}^{k} \rangle + r \langle \wb{R}^{k-1} -\wb{P}^{k},
  \wb{P}^{k} \rangle  \\
  + \lambda \langle \wb{b}^{k-1}, \wb{Q}^{k} \rangle + \lambda\langle
  \wb{P}^{k} - \wb{Q}^{k}, \wb{Q}^{k} \rangle + r \langle \wb{d}^{k-1},
  \wb{R}^{k} \rangle + r \langle \wb{P}^{k} - \wb{R}^{k}, \wb{R}^{k}
  \rangle \geq 0.
\end{multline}
This gives us, after organizing terms
\begin{multline}
  -\lambda\langle \wb{b}^{k-1}, \wb{P}^{k} -\wb{Q}^{k} \rangle - \lambda \|\wb{Q}^{k} - \wb{P}^{k}\|^2 - \lambda\langle \wb{P}^{k} ,  \wb{Q}^{k} -\wb{Q}^{k-1} \rangle  \\
  -r\langle \wb{d}^{k-1}, \wb{P}^{k} -\wb{R}^{k} \rangle - r
  \|\wb{R}^{k} - \wb{P}^{k}\|^2 - r \langle \wb{P}^{k},
  \wb{R}^{k} -\wb{R}^{k-1} \rangle \geq 0
\end{multline}
Combining the above inequality with \eqref{eqn:bderror}, we have
\begin{equation}\label{eqn:Monotone1}
  \begin{aligned} 
    \bigl(\lambda \|\wb{b}^{k-1}\|^2  + r \|\wb{d}^{k-1}\|^2 \bigr) & -   \bigl(\lambda \|\wb{b}^{k}\|^2  + r \|\wb{d}^{k}\|^2 \bigr)  \\
    & = \lambda \bigl(-2\langle \wb{b}^{k-1}, \wb{P}^{k} - \wb{Q}^{k} \rangle - \|\wb{P}^{k} - \wb{Q}^{k}\|^2\bigr) \\
    & \qquad + r \bigl(-2\langle \wb{d}^{k-1}, \wb{P}^{k} - \wb{R}^{k} \rangle - \|\wb{P}^{k} - \wb{R}^{k}\|^2 \bigr)  \\
    & \geq \lambda \|\wb{Q}^{k} - \wb{P}^{k}\|^2 + 2 \lambda\langle \wb{P}^{k},  \wb{Q}^{k} -\wb{Q}^{k-1}\rangle \\
    & \qquad + r \|\wb{R}^{k} - \wb{P}^{k}\|^2 + 2 r \langle
    \wb{P}^{k}, \wb{R}^{k} -\wb{R}^{k-1} \rangle
  \end{aligned}
\end{equation}

Now, we calculate $  \langle \wb{P}^{k} ,  \wb{Q}^{k} -\wb{Q}^{k-1}\rangle$. It is clear that
\begin{multline}\label{eqn:ErrorPQ1}
  \langle \wb{P}^{k} ,  \wb{Q}^{k} -\wb{Q}^{k-1}\rangle = \langle \wb{P}^{k} - \wb{P}^{k-1} ,  \wb{Q}^{k} -\wb{Q}^{k-1}\rangle +  \langle \wb{P}^{k-1} - \wb{Q}^{k-1} ,  \wb{Q}^{k} -\wb{Q}^{k-1}\rangle \\
  + \langle \wb{Q}^{k-1} , \wb{Q}^{k} -\wb{Q}^{k-1}\rangle
\end{multline}
Note that $\displaystyle Q^{k-1} = \arg\min_{Q} g(Q) + \frac{\lambda}{2} \|Q - P^{k-1} - b^{k-2} \|^2 $. Thus, for any $Q\in \RR^{n\times n}$, we have 
\begin{eqnarray}
g(Q) - g(Q^{k-1}) + \lambda\langle Q^{k-1} - P^{k-1} - b^{k-2}, Q - Q^{k-1} \rangle \geq 0 
\end{eqnarray}
In particular, let $Q = Q^{k}$, we have
\begin{eqnarray}  \label{eqn:gQkk1}
  g(Q^{k}) - g(Q^{k-1}) + \lambda\langle Q^{k-1} - P^{k-1} - b^{k-2}, Q^{k} - Q^{k-1} \rangle \geq 0 
\end{eqnarray}
On the other hand, set $Q = Q^{k-1}$ in \eqref{eqn:gQk}, we get
\begin{eqnarray}  \label{eqn:gQkk2}
g(Q^{k-1}) - g(Q^{k}) + \lambda\langle Q^{k}- P^{k} - b^{k-1}, Q^{k-1} - Q^{k} \rangle \geq 0  
\end{eqnarray}

The summation of \eqref{eqn:gQkk1} and \eqref{eqn:gQkk2} yields
\begin{equation}
\langle b^{k-1} - b^{k-2}, Q^{k} - Q^{k-1} \rangle + \langle Q^{k-1} - Q^{k} + P^{k} - P^{k-1} , Q^{k} - Q^{k-1} \rangle \geq 0 
\end{equation}
Note that $P^{k} - P^{k-1} = \wb{P}^{k} - \wb{P}^{k-1}, Q^{k} - Q^{k-1} = \wb{Q}^{k} - \wb{Q}^{k-1}, b^{k-1} - b^{k-2} = \wb{P}^{k-1} - \wb{Q}^{k-1}$, 
thus we have
\begin{equation}
\langle \wb{P}^{k-1} - \wb{Q}^{k-1}, \wb{Q}^{k} - \wb{Q}^{k-1} \rangle + \langle \wb{P}^{k} - \wb{P}^{k-1} , \wb{Q}^{k} - \wb{Q}^{k-1} \rangle \geq \|\wb{Q}^{k} - \wb{Q}^{k-1}\|^2 
\label{eqn:ErrorPQ2}
\end{equation}

Combine \eqref{eqn:ErrorPQ2} with \eqref{eqn:ErrorPQ1}, we have
\begin{equation}
  \langle \wb{P}^{k}, \wb{Q}^{k} - \wb{Q}^{k-1} \rangle \geq \| \wb{Q}^{k} - \wb{Q}^{k-1} \|^2 + \langle \wb{Q}^{k-1}, \wb{Q}^{k} - \wb{Q}^{k-1} \rangle 
\label{eqn:ErrorPQ3}
\end{equation}
Similarly, we have
\begin{equation}
  \langle \wb{P}^{k}, \wb{R}^{k} - \wb{R}^{k-1} \rangle \geq \| \wb{R}^{k} - \wb{R}^{k-1} \|^2 + \langle \wb{R}^{k-1}, \wb{R}^{k} - \wb{R}^{k-1} \rangle 
\label{eqn:ErrorPR3}
\end{equation}

Substitute \eqref{eqn:ErrorPQ3} and \eqref{eqn:ErrorPR3} into
\eqref{eqn:Monotone1}, we have
\begin{equation} \label{eqn:Monotone2}
  \begin{aligned}
    \bigl(\lambda \|\wb{b}^{k-1}\|^2  + r \|\wb{d}^{k-1}\|^2 ) & -  (\lambda \|\wb{b}^{k}\|^2  + r \|\wb{d}^{k}\|^2 )   \\
    & \geq  \lambda \|\wb{Q}^{k} - \wb{P}^{k}\|^2 + 2 \lambda\langle \wb{P}^{k},  \wb{Q}^{k} - \wb{Q}^{k-1}\rangle  \\
    &  \qquad + r \|\wb{R}^{k} - \wb{P}^{k}\|^2 + 2  r \langle \wb{P}^{k} ,  \wb{R}^{k} - \wb{R}^{k-1} \rangle  \\
    & \geq \lambda \|\wb{Q}^{k} - \wb{P}^{k}\|^2 + 2 \lambda
    \bigl(\| \wb{Q}^{k} - \wb{Q}^{k-1} \|^2 + \langle \wb{Q}^{k-1}, \wb{Q}^{k} - \wb{Q}^{k-1} \rangle \bigr)  \\
    & \qquad + r \|\wb{R}^{k} - \wb{P}^{k}\|^2 + 2  r \bigl(\| \wb{R}^{k} - \wb{R}^{k-1} \|^2 + \langle \wb{R}^{k-1}, \wb{R}^{k} - \wb{R}^{k-1} \rangle\bigr) \\
    &=  \lambda \|\wb{Q}^{k} - \wb{P}^{k}\|^2 +  \lambda \bigl(\| \wb{Q}^{k} \|^2  - \| \wb{Q}^{k-1} \|^2  + \| \wb{Q}^{k} - \wb{Q}^{k-1} \|^2  \bigr)  \\
    & \qquad + r \|\wb{R}^{k} - \wb{P}^{k}\|^2 + r \bigl(\|
    \wb{R}^{k} \|^2 - \| \wb{R}^{k-1} \|^2 + \| \wb{R}^{k} -
    \wb{R}^{k-1} \|^2 \bigr)
\end{aligned}
\end{equation}
which yields
\begin{multline} \label{eqn:Monotone3}
  \bigl(\lambda \|\wb{b}^{k-1}\|^2  + r \|\wb{d}^{k-1}\|^2 + \lambda \| \wb{Q}^{k-1} \|^2 + r \| \wb{R}^{k-1} \|^2 \bigr) \\
  -  \bigl(\lambda \|\wb{b}^{k}\|^2  + r \|\wb{d}^{k}\|^2 +  \lambda \| \wb{Q}^{k} \|^2 + r \| \wb{R}^{k} \|^2 \bigr)    \\
  \geq \lambda \|\wb{Q}^{k} - \wb{P}^{k}\|^2 + \lambda \| \wb{Q}^{k} -
  \wb{Q}^{k-1} \|^2 + r \|\wb{R}^{k} - \wb{P}^{k}\|^2 + r \|
  \wb{R}^{k} - \wb{R}^{k-1} \|^2
\end{multline}

This concludes that the sequence $\Big\{ \lambda \|\wb{b}^{k}\|^2 + r
\|\wb{d}^{k}\|^2 + \lambda \| \wb{Q}^{k} \|^2 + r \| \wb{R}^{k} \|^2
\Big\}_k$ is non-increasing and hence convergent.
This further implies,
\renewcommand{\labelenumi}{(\alph{enumi})}
\begin{enumerate}
\item $\{P^k\}_k, \{Q^k\}_k, \{R^k\}_k, \{b^k\}_k, \{d^k\}_k$ are all bounded sequences, and hence the sequences has limit points. 
\item $\lim_{k\rightarrow\infty} \| Q^k - P^k\| = 0$ and $\lim_{k\rightarrow\infty} \| R^k - P^k\| = 0$.
\end{enumerate}
% Note that \eqref{eqn:Monotone3} implies that for $m < n$
% \begin{equation}
%   \begin{aligned}
%     \lambda \norm{\wb{Q}^n - \wb{Q}^m}  & \leq \sum_{k= m+1}^n  \Bigl[  \bigl(\lambda \|\wb{b}^{k-1}\|^2  + r \|\wb{d}^{k-1}\|^2 + \lambda \| \wb{Q}^{k-1} \|^2 + r \| \wb{R}^{k-1} \|^2 \bigr) \\
%     & \hspace{6em} - \bigl(\lambda \|\wb{b}^{k}\|^2 + r
%     \|\wb{d}^{k}\|^2 + \lambda \| \wb{Q}^{k} \|^2 + r \| \wb{R}^{k} \|^2 \bigr) \Bigr] \\
%     & = \bigl(\lambda \|\wb{b}^{m}\|^2  + r \|\wb{d}^{m}\|^2 + \lambda \| \wb{Q}^{m} \|^2 + r \| \wb{R}^{m} \|^2 \bigr) \\
%     & \hspace{6em} - \bigl(\lambda \|\wb{b}^{n}\|^2 + r
%     \|\wb{d}^{n}\|^2 + \lambda \| \wb{Q}^{n} \|^2 + r \| \wb{R}^{n} \|^2 \bigr) 
%   \end{aligned}
% \end{equation}
% Hence, $\{Q^k\}_k$ is a Cauchy sequence, and similarly $\{R^k\}_k, \{b^k\}_k, \{d^k\}_k$ as well, and hence $\{P^k\}_k$. 
Therefore, the sequences have limit points. Let us denote $(\wt{P},
\wt{Q}, \wt{R}; \wt{b}, \wt{d})$ as a limit point, that is, a
subsequence converges 
\begin{equation}
  \lim_{j\to \infty} (P^{k_j}, Q^{k_j}, R^{k_j}; b^{k_j}, d^{k_j}) 
  = (\wt{P}, \wt{Q}, \wt{R}; \wt{b}, \wt{d}). 
\end{equation}
We now prove that $(\wt{P}, \wt{Q}, \wt{R})$ is a minimum of the variational problem \eqref{eq:three}, \textit{i.e.}
\begin{equation}\label{eq:limitPQR}
  f(\wt{P}) + g(\wt{Q}) +  h(\wt{R}) = \lim_{j\rightarrow \infty} f(P^{k_j}) + g(Q^{k_j}) + h(R^{k_j}) = f(\Ps) + g(\Qs) + h(\Rs)
\end{equation}

First note that since $(\Ps, \Qs, \Rs; \bs, \ds)$ is a saddle point, we have
\begin{multline}
f(\Ps) + g(\Qs) + h(\Rs) \leq f(P^{k_j}) + g(Q^{k_j}) + h(R^{k_j}) + \frac{\lambda}{2}\|P^{k_j} - Q^{k_j}\|^2 \\
+ \lambda\langle P^{k_j} - Q^{k_j}, \bs \rangle + \frac{r}{2}\|P^{k_j} - R^{k_j}\|^2 + r\langle P^{k_j} - R^{k_j}, \ds \rangle
\end{multline}
Taking the limit $j \to \infty$, we get 
\begin{equation}
  f(\Ps) + g(\Qs) + h(\Rs) \leq f(\wt{P}) + g(\wt{Q}) +  h(\wt{R}). 
\end{equation}

On the other hand, taking $P = \Ps$, $Q = \Qs$, and $R = \Rs$ in \eqref{eqn:fPk}--\eqref{eqn:hRk}, we get
\begin{align*}
  f(\Ps) & + g(\Qs) + h(\Rs) \\
  & \geq   f(P^{k_j}) + g(Q^{k_j}) + h(R^{k_j})  - \lambda\langle P^{k_j} - Q^{k_j-1} + b^{k_j-1}, \Ps - P^{k_j} \rangle  \\
  & \qquad - r \langle P^{k_j} - R^{k_j} + d^{k_j-1}, \Ps - P^{k_j} \rangle - \lambda\langle Q^{k_j}- P^{k_j} - b^{k_j-1}, \Qs - Q^{k_j} \rangle \\
  & \qquad - r \langle R^{k_j} - P^{k_j} - d^{k_j-1}, \Rs - R^{k_j}
  \rangle  \\
  &= f(P^{k_j}) + g(Q^{k_j}) + h(R^{k_j}) \\
  &  \qquad - \lambda\langle  b^{k_j-1}, Q^{k_j} - P^{k_j} \rangle - \lambda\langle P^{k_j} - Q^{k_j-1} , \Ps - P^{k_j} \rangle - \lambda\langle Q^{k_j}- P^{k_j} , \Qs - Q^{k_j} \rangle  \\
  & \qquad - r \langle d^{k_j-1}, R^{k_j} - P^{k_j} \rangle - r \langle
  P^{k_j} - R^{k_j} , \Ps - P^{k_j} \rangle - r \langle R^{k_j} -
  P^{k_j} , \Rs - R^{k_j} \rangle
\end{align*}
From \eqref{eqn:Monotone3}, we have $\{P^{k_j}\}, \{Q^{k_j}\}, \{R^{k_j}\}, \{b^{k_j}\}, \{d^{k_j}\}$ are all bounded sequences, and furthermore, 
\begin{eqnarray}
\lim_{j\rightarrow\infty} \| Q^{k_j }- P^{k_j}\| = 0, \qquad \lim_{j\rightarrow\infty} \| Q^{k_j }- Q^{k_j - 1}\| = 0.\nonumber \\
\lim_{j\rightarrow\infty} \| R^{k_j} - P^{k_j}\| = 0, \qquad \lim_{j\rightarrow\infty} \| R^{k_j }- R^{k_j - 1}\| = 0. \nonumber
\end{eqnarray}
Taking the limit $j \to \infty$, we then get 
\begin{equation}
  f(\Ps) + g(\Qs) + h(\Rs) \geq f(\wt{P}) + g(\wt{Q}) +  h(\wt{R}). 
\end{equation}
Hence, the limit point is a minimizer of the variational principle. 

Finally, repeating the derivation of \eqref{eqn:Monotone3} by replacing $(\Ps, \Qs, \Rs)$ by $(\wt{P}, \wt{Q}, \wt{R})$, we get convergence of the whole sequence due to the monotonicity. 
\end{proof}

%\bibliographystyle{amsxport}
%% argument is your BibTeX string definitions and bibliography database(s)
%\bibliography{compressedmodes,jl}

% \bib, bibdiv, biblist are defined by the amsrefs package.

\end{document}